\documentclass[envcountsame,runningheads]{llncs}

\usepackage{hyperref} 
\usepackage{amssymb}
\usepackage{amsmath}
\usepackage{temporal}
\usepackage{enumerate}
\usepackage{array}      

\usepackage{bold-extra}

\usepackage{tikz}
  \usetikzlibrary{automata,positioning}

\usepackage{caption}
\usepackage{subfigure}

\usepackage{enumitem}
  \setlist{nolistsep}

\newif\ifnotfinal

\newif\ifprog

\newif\ifall
\newif\ifblack

\ifblack
  \newcommand{\blackout}[1]{{\color{black!50} #1}}  
\else
  \newcommand{\blackout}[1]{}
\fi

\usepackage[textsize=small,backgroundcolor=red,disable]{todonotes}
  \presetkeys{todonotes}{inline}{}
  
\let\note\relax

\ifnotfinal
  \usepackage{trackchanges} 
\else
  \usepackage[finalold]{trackchanges} 
\fi

\addeditor{SR}
\addeditor{SJ}
\addeditor{AK}
\addeditor{BA}
\addeditor{RB}

\makeatletter
\renewcommand\tableofcontents{%
\@starttoc{toc}%
}
\makeatother

\makeatletter
  \let\LLNCSproof\proof                 
  \let\LLNCSendproof\endproof         
  \let\proof\@undefined                  
  \let\endproof\@undefined              
\makeatother

\usepackage{amsthm}

\makeatletter
  \let\proof\LLNCSproof                         
  \let\endproof\LLNCSendproof              
\makeatother

\makeatletter
  \newtheorem*{rep@theorem}{\rep@title}
  \newcommand{\newreptheorem}[2]{%
  \newenvironment{rep#1}[1]{%
   \def\rep@title{#2 \ref{##1}}%
   \begin{rep@theorem}}%
   {\end{rep@theorem}}}
\makeatother

\newreptheorem{theorem}{Theorem}
\newreptheorem{lemma}{Lemma}

\newtheorem{observation}[theorem]{Observation}

\newif\iffig
\IfFileExists{ignore_figures}{\figfalse}{\figtrue}

\renewcommand{\P}{\mathcal{P}}

\newcommand{\ak}[1]{\note[AK]{#1}}
\newcommand{\sj}[1]{\note[SJ]{#1}}
\newcommand{\sr}[1]{\note[SR]{#1}}

\newcommand{\rb}[1]{\note[RB]{#1}}

\newcommand\destut{\mathsf{destutter}}

\newcommand\tokengraph[2]{{{#1}}|{#2}}

\newcommand\abstraction[2]{{{#1}}|{#2}}

\newcommand\dcon[2]{\textsc{con}_d\tokengraph{#1}{#2}}

\newcommand{\trans}[3]{#1 \stackrel{\mathsf{#3}}{\rightarrow} #2}
\newcommand{\Trans}{\delta}
\newcommand{\LabelFun}{\lambda}
\newcommand{\AP}{\textsf{AP}}
\newcommand{\APproc}{\textsf{AP}_{\textsf{pr}}}
\newcommand{\APsys}{\textsf{AP}_{\textsf{sys}}}
\newcommand{\Pproctemp}{{\mathcal{P}}}
\newcommand{\Psimptok}{\Pproctemp_{\textsf {\textit{u}}}}
\newcommand{\Psnd}{\Pproctemp_{\textsf {\textit{snd}}}}
\newcommand{\Prcv}{\Pproctemp_{\textsf {\textit{rcv}}}}
\newcommand{\Prcvsnd}{\Pproctemp_{\textsf{\textit{sndrcv}}}}
 
\newcommand\IndSet{\mathbb{N}}
\newcommand{\ActionsInt}{\Sigma_{\mathsf{int}}}
\newcommand{\ActionsProc}{\Sigma_{\mathsf{pr}}}
\newcommand{\Locals}{Q}
\newcommand{\LocalsI}{\Locals_0}
\newcommand{\trcv}{\mathsf{rcv}}
\newcommand{\tsnd}{\mathsf{snd}}
\newcommand{\dirset}{\mathsf{Dir}}

\newcommand{\dir}{\mathsf{dir}}
\newcommand{\tok}{\mathsf{tok}}

\newcommand{\pspec}{\varphi}
\newcommand{\Pspec}{{\mathcal F}}

\def\AQLTL{AQLTL}

\newcommand{\PMCP}{{\textsf{PMCP}}}

\newcommand{\pgraph}{\ensuremath{\mathbf {G}}}
\newcommand{\pgraphH}{\ensuremath{\mathbf {H}}}

\newcommand{\pring}{\ensuremath{\mathbf {R}}}

\newcommand\ltlF{\ensuremath{\eventually}}

\newcommand\ltlG{\always}
\newcommand\ltlU{\until}

\newcommand\CTLstar{\ensuremath{\textsf{CTL}^\ast}}

\newcommand\CTLmX{\ensuremath{\textsf{CTL}\backslash\textsf{X}}}
\newcommand\CTLstarmX{\ensuremath{\textsf{CTL}^\ast\backslash\textsf{X}}}
\newcommand\LTL{{\textsf{LTL}}}
\newcommand\LTLmX{\ensuremath{\textsf{LTL}\backslash\textsf{X}}}

\newcommand\TL{\ensuremath{\textsf{TL}}}
\newcommand\piTL{\ensuremath{\{\forall,\exists\}^\ast{\textsf{-TL}}}}
\newcommand\kTL{\ensuremath{\{\forall,\exists\}^k{\textsf{-TL}}}}

\newcommand\kLTLmX{\ensuremath{\{\forall,\exists\}^k{\textsf{-LTL}\backslash\textsf{X}}}}

\newcommand\piCTLstarmX{\ensuremath{\{\forall,\exists\}^*\textsf{-CTL}^\ast\backslash\textsf{X}}}

\newcommand\universalpiCTLstarmX{\ensuremath{\{\forall\}^*\textsf{-CTL}^\ast\backslash\textsf{X}}}

\newcommand\twoCTLmX{\ensuremath{\{\forall,\exists\}^2\textsf{-CTL}\backslash\textsf{X}}}

\newcommand\twoallCTLstarmX{\ensuremath{\forall\forall\textsf{-CTL}^\ast\backslash\textsf{X}}}
\newcommand\twoexistsCTLstarmX{\ensuremath{\exists\exists\textsf{-CTL}^\ast\backslash\textsf{X}}}

\newcommand\kCTLstarmX{\ensuremath{\{\forall,\exists\}^k\textsf{-CTL}^*\backslash\textsf{X}}}
\newcommand\allkCTLstarmX{\ensuremath{\{\forall\}^k\textsf{-CTL}^*\backslash\textsf{X}}}

\newcommand\pidCTLstarmX{\ensuremath{\{\forall,\exists\}^*\textsf{-CTL}_d^*\backslash\textsf{X}}}

\newcommand\twoallCTLmX{\ensuremath{\{\exists\}^2\textsf{-CTL}\backslash\textsf{X}}}

\newcommand\dkCTLstarmX{\ensuremath{\{\forall,\exists\}^k\textsf{-CTL}^\ast_d\backslash\textsf{X}}}
\newcommand\xCTLstarmX[1]{\ensuremath{\textsf{CTL}^\ast_{#1}\backslash\textsf{X}}}
\newcommand\dCTLstarmX{\xCTLstarmX{d}}

\newcommand\nineexistsCTLmX{\ensuremath{\{\exists\}^9\textsf{-CTL}\backslash\textsf{X}}}

\newcommand\piLTLmX{\ensuremath{\{\forall,\exists\}^*\textsf{-LTL}\backslash\textsf{X}}}
\newcommand\nineallLTLmX{\ensuremath{\{\forall\}^9\textsf{-LTL}\backslash\textsf{X}}}

\newcommand\oneSLTLmX{{\textsf{1-SLTL}}\backslash\textsf{X}}

\newcommand\VarNames{\textsf{Vars}}

\newcommand\Nat{\mathbb{N}}

\newcommand{\card}[1]{{|}#1{|}}

\newcommand{\ValidRun}{\textit{VALID}_{x,\overline{y}}}

\newcommand{\specialcell}[2][c]{%
  \begin{tabular}[#1]{@{}c@{}}#2\end{tabular}}

\newcommand{\ourtitle}{Parameterized Model Checking \\ of Token-Passing Systems}

\title{\ourtitle}
\author{Benjamin Aminof$^1$, Swen Jacobs$^2$, Ayrat Khalimov$^2$, Sasha Rubin$^{1,3}$
\thanks{This work was supported by 
the Austrian Science Fund through grant P23499-N23 and
through the RiSE network (S11403, S11405, S11406, S11407-N23);
ERC Starting Grant (279307: Graph Games);
Vienna Science and Technology Fund (WWTF) grants PROSEED, ICT12-059, and VRG11-005.}}
\institute{$^1$IST Austria ({\sf first.last@ist.ac.at}), 
$^2$TU Graz ({\sf first.last@iaik.tugraz.at}),
$^3$TU Wien}
\authorrunning{B. Aminof, S. Jacobs, A. Khalimov, S. Rubin}
\titlerunning{\ourtitle}

\begin{document}
\maketitle

\begin{abstract}
We revisit the parameterized model checking problem for token-passing systems
and specifications in indexed $\textsf{CTL}^\ast \backslash \textsf{X}$. Emerson and Namjoshi (1995, 2003) have shown that parameterized model checking of indexed
$\textsf{CTL}^\ast \backslash \textsf{X}$ in uni-directional token rings can be reduced to
checking rings up to some \emph{cutoff} size. Clarke et al. (2004) have shown a similar
result for general topologies and indexed $\textsf{LTL} \backslash \textsf{X}$, provided
processes cannot choose the directions for sending or receiving the token.

~~~~We unify and substantially extend these results by systematically exploring fragments of indexed $\textsf{CTL}^\ast
\backslash \textsf{X}$ with respect to general topologies. For each
fragment we establish whether a cutoff exists, and for some concrete topologies, such as rings, cliques and stars, we infer small cutoffs. Finally, we show that the problem becomes undecidable, and thus no cutoffs exist, if processes are allowed to choose the directions in which they send or from which they receive the token.

\blackout{
The parameterized model checking problem for concurrent systems composed of identical
processes is to decide if a given temporal logic specification holds irrespective of the number of
processes. We address the parameterized model checking problem for token-passing systems
and specifications in indexed $\textsf{CTL}^\ast \backslash \textsf{X}$.


\
%

Previous results have shown that parameterized model checking can sometimes be reduced to model checking systems with a number of processes up to some {\em cutoff} value.  Notably, for token-passing systems, such cutoffs are known for uni-directional ring topologies with respect to the prenex fragment of indexed  $\textsf{CTL}^\ast \setminus \textsf{X}$ (Emerson and Namjoshi 1995), and for general topologies with respect to the prenex fragment of indexed $\textsf{LTL} \setminus \textsf{X}$ provided processes are not allowed to choose the direction to which the token is sent or from which it is received (Clarke, Talupur, Touili and Veith 2004).

\

We complete the picture by systematically exploring indexed $\textsf{CTL}^\ast \setminus \textsf{X}$ and its fragments with respect to general topologies, and provide for each fragment either a cutoff or a proof that no cutoff exists. We unify and substantially extend the cutoff results mentioned above. In almost all cases where a cutoff does not exist we show that the parameterized model checking problem is actually undecidable. We also consider how the problem changes if processes are allowed to choose the directions in which they send or receive the token.

}

\end{abstract}

\section{Introduction}

As executions of programs and protocols are increasingly distributed
over multiple CPU cores or even physically separated computers,
correctness of concurrent systems is one of the primary challenges of
formal methods today.  Many concurrent systems consist of an arbitrary
number of identical processes running in parallel. The parameterized
model checking problem (PMCP)  
for concurrent systems is 
to decide if a given temporal logic specification holds 
irrespective of the number of participating processes. 

The PMCP is undecidable in many cases. For example, it is undecidable already for safety specifications and finite-state processes communicating by passing a binary-valued token around a uni-directional ring~\cite{Suzuki,EN95}. However, decidability may be regained by restricting the communication primitives, the topologies under consideration (i.e., the underlying graph describing the communication paths between the processes), or the specification language. In particular, previous results have shown that parameterized model checking can sometimes be reduced to model checking a finite number of instances of the system, up to some {\em cutoff} size.


For token-passing systems (TPSs) with uni-directional ring topologies,
such cutoffs are known for specifications in the prenex fragment of
indexed \CTLstar\ without the next-time operator (\CTLstarmX)~\cite{EN95popl,EN95}.\sr{What do you think of this footnote?  It is common in all these results (as well as this work) to consider specification languages
without $X$ (the "next" operator). The basic reason is that since we discretize time, when a process makes an internal transition time has to proceed by one step. In particular, if a formula only talks about certain processes, then when some other processes make internal moves, time can advance any number of steps. But this advance, which could be expressed by a formula that uses $X$, should not be allowed to effect the truth of the formula.}
\sr{or this...  It is common in all these results (as well as this work) to consider specification languages
without $X$ (the "next" operator). The basic reason is that by
considering discreet time, one has to assume a new time step whenever
any process in the system changes its internal state. Thus,  event
that happened exactly one time step after another event when the
number of processes is small, is considered to happen a few time steps
later when more processes are involved, simply because one of the new
processes made some internal move. Since this should usually not
affect the validity of the formula, the use of ``next'' is
disallowed, and is replaced by ``eventually'', or by
``until'' when one wishes to express the requirement the some
configuration q is entered immediately after a configuration p.}
\sj{I think it should be much shorter: ``It is common in all these
  results to consider specifications without ${\sf X}$, as these are
  stuttering-insensitive. This is desirable since the execution model
  is in general asynchronous, and we are interested in specifications
  that hold regardless of the number of components.''}
For token-passing in 
general topologies, cutoffs are known for the prenex fragment of
indexed $\LTLmX$, provided that processes are not
allowed to choose the direction to which the token is sent or from
which it is received~\cite{CTTV04}. In this paper we generalize these
results and elucidate what they have in common. 
\medskip
\noindent{\textbf{Previous Results.}}
In their seminal paper~\cite{EN95}, Emerson and Namjoshi consider systems where the token does not carry messages, and specifications are in prenex indexed temporal logic --- i.e., quantifiers $\forall$ and $\exists$ over processes appear in a block at the front of the formula.
They use the important concept of a {\em cutoff} --- a number $c$ such
that the PMCP for a given class of systems and specifications can be
reduced to model checking systems with up to $c$
processes.\sr{rephrase?: They present results about  cutoffs --- a
  number $c$ such that a PMCP can be reduced to model checking systems
  with up to $c$ processes.} If model checking is decidable, then
existence of a cutoff implies that the PMCP is
decidable.\sj{Should we also say here ``if MC is decidable'' instead
  of ``if processes are finite''?} Conversely, if the PMCP is undecidable, then there can be no
cutoff for such systems.

For uni-directional rings, Emerson and Namjoshi provide cutoffs for
formulas with a small number $k$ of quantified index variables, and
state that their proof method allows one to obtain cutoffs for other
quantifier prefixes. In brief, cutoffs exist for the branching-time
specification language prenex indexed \CTLstarmX\ and the highly regular topology of uni-directional rings.

Clarke et al.~\cite{CTTV04} consider the PMCP for token-passing systems arranged in general topologies.  Their main result is that the PMCP for systems with arbitrary topologies and $k$-indexed $\LTLmX$ specifications (i.e., specifications with $k$ quantifiers over processes in the prenex of the formula) can be reduced to combining the results of model-checking finitely many topologies of size at most $2k$ \cite[Theorem $4$]{CTTV04}. Their proof implies that, for each $k$, the PMCP for linear-time specifications in $k$-indexed \LTLmX\ and general topologies has a cutoff. 
\begin{table}[b]
\vspace{-0.4cm}
\parbox{.49\linewidth}{
\centering
\input{map_du}
\caption{Direction-Unaware TPSs.}
\label{fig:DU}
}
\hfill
\parbox{.49\linewidth}{
\centering
\input{map_da}
\caption{Direction-Aware TPSs.}
\label{fig:DA}
}
\vspace{-0.4cm}
\end{table}

\smallskip
\noindent{\textbf{Questions.}}
Comparing these results, an obvious question is: are there cutoffs for
branching time temporal logics and arbitrary topologies (see
Table~\ref{fig:DU})? Clarke et al. already give a first
answer~\cite[Corollary $3$]{CTTV04}. They prove that there is no
cutoff for token-passing systems with arbitrary topologies and
specifications from $2$-indexed \CTLmX. However, their proof makes use
of formulas with unbounded nesting-depth of path quantifiers. This lead us to the first question.

\medskip
\noindent
\emph{Question $1$.} Is there a way to stratify $k$-indexed \CTLstarmX\ such that for each level of the stratification there is a cutoff for systems with arbitrary topologies? In particular, does stratification by nesting-depth of path quantifiers do the trick?

\medskip
\noindent Cutoffs for $k$-indexed temporal logic fragments immediately yield that for each $k$ there is an algorithm (depending on $k$) for deciding the PMCP for $k$-indexed temporal logic. However, this does not imply that there is an algorithm that can compute the cutoff for a given $k$. In particular, it does not imply that PMCP for full prenex indexed temporal logic is decidable.

\medskip
\noindent
\emph{Question $2$.}  For which topologies (rings, cliques, all?) can one conclude that the PMCP for the full prenex indexed temporal logic is decidable?

\medskip
\noindent
Finally, an important implicit assumption in Clarke et al.~\cite{CTTV04} is that processes are not {\em direction aware}, i.e., they cannot sense or choose in which direction the token is sent, or from which direction it is received. In contrast, Emerson and Kahlon~\cite{EK04} show that cutoffs exist for certain direction-aware systems in bi-directional rings (see Section~\ref{sec:related}). We were thus motivated to understand to what extent existing results about cutoffs can be lifted to direction-aware systems, see Table~\ref{fig:DA}.

\medskip
\noindent
\emph{Question $3$.} Do cutoffs exist for direction-aware systems on arbitrary topologies and $k$-indexed temporal logics (such as \LTLmX\ and \CTLmX)?

\blackout{
\todo{Old intro+questions}
Comparing these results, one obvious question is: are there cutoffs for branching temporal logics and arbitrary topologies? Clarke et al. already give a first answer. They prove that 
there is no cutoff for token-passing systems with arbitrary topologies and specifications from $2$-indexed \CTLmX \cite[Corollary $3$]{CTTV04}.  The proof provides a way to construct, for every number $n$, a \sr{satisfiable} formula that holds in a system only if the topology has size at least $n$. Thus, there cannot be a cutoff.

An important assumption in Clarke et al. is that processes are not {\em direction aware}, i.e., they cannot sense or choose in which direction the token is sent, or from which direction it is received. In contrast, Emerson and Kahlon~\cite{EK04} show that cutoffs exist for certain direction-aware systems in bi-directional rings, see Section~\ref{sec:related}.\footnote{In uni-directional rings the notion is irrelevant because every process has exactly one outgoing direction and exactly one incoming direction.} It is an interesting question whether the results of Clarke et al. can be extended to direction-aware systems.\sr{Suggestion to replace second+third sentences: In contrast, cite[Theorem XX]{EK04} prove cutoffs exist for certain direction-aware systems in bi-directional rings, see Section~ref{sec:related}. We were curious if the results of Clarke et al., and Emerson and Namjoshi, could be extended to direction-aware systems.\\ SJ: I implemented some of the changes, but I don't like the ``we were curious'' part - can we leave it like this?\\ AK: i like it, especially `curious'\\SR. 'it is an interesting question' is too weak, especially if we don't say why it is interesting. How about:  We were thus motivated to understand to what extent these existing results can be lifted to branching-time specifications on arbitrary topologies with direction-aware systems. or: We were curious to what extent the CTTV04 results relied on the processes being direction-unaware.}

\subsubsection{Motivating questions.} \sr{shall we call these 'motivating problems' to distinguish them from the questions above? \\ AK: `Challenges'?}
\todo{The point about these questions is that they should follow from the discussion above. Every reader that is paying attention should say 'Yes, these are the right questions to ask given the previous work'.}
The considerations above lead to the following questions:
\todo{The questions need additional explanation to make them meaningful. It is there in Q1, but looks to complicated. Q2 also needs an explanation (why this is a different problem from cutoffs for all $k$). Q3 may be OK. Can we find a different structure? Maybe not separate Q from A, but have three paragraphs were we give Q and A together, in each?}

\emph{Question $1$.} Is there a way to stratify $k$-indexed \CTLstarmX\ such that, for each level of the stratification, there is a cutoff for direction-unaware systems with arbitrary topologies?\\ In particular, the construction of the formulas in \cite[Corollary $3$]{CTTV04} is such that with increasing $n$, the nesting depth of path quantifiers grows unboundedly. So, does the stratification by nesting depth of path quantifiers do the trick?


\emph{Question $2$.} For which topologies is the PMCP decidable for specifications from prenex indexed \CTLstarmX\ and direction-unaware systems? 


\emph{Question $3$.} Do cutoffs exist for direction-aware systems on arbitrary topologies and $k$-indexed temporal logics (such as \LTLmX)?
}

\medskip
\noindent
\textbf{Our contributions.} 
In this paper, we answer the questions above, unifying and substantially extending the known cutoff results:

\medskip
\noindent
\emph{Answer to Question $1$.}
Our main positive result (Theorem~\ref{thm:cutoff_kctld}) states that for arbitrary parameterized topologies $\pgraph$ there is a cutoff for specifications in $k$-indexed \dCTLstarmX\ --- the cutoff depends on $\pgraph$, the number $k$ of the process quantifiers, and the nesting depth $d$ of path quantifiers. In particular, indexed \LTLmX\ is included in the case $d = 1$, and so our result generalizes the results of Clarke et al.~\cite{CTTV04}. 
 

\smallskip
\noindent
\emph{Answer to Question $2$.}
We prove (Theorem~\ref{thm:undecprenex}) that there exist topologies for
which the PMCP is undecidable for specifications in prenex indexed
\CTLmX\ or \LTLmX. Note that this 
undecidability result does not contradict the existence of cutoffs
(Theorem~\ref{thm:cutoff_kctld}), since cutoffs may not be computable
from $k,d$ (see the note on decidability in
Section~\ref{sec:definition-pmcp}).
However, for certain topologies our positive result is constructive and we can
compute cutoffs given $k$ and $d$
(Theorem~\ref{thm:explicit_cutoffs}). To illustrate, we show that
rings have a cutoff of $2k$, cliques of $k+1$, and stars of $k+1$
(independent of $d$). In particular, PMCP is decidable for these
topologies and specifications in prenex indexed \CTLstarmX. 

\medskip
\noindent
\emph{Answer to Question $3$.}
The results just mentioned assume that processes are not direction-aware.  Our main negative result (Theorem~\ref{thm:DA}) states that if processes can control at least one of the directions (i.e., choose in which direction to send or from which direction to receive) then the PMCP for arbitrary topologies and $k$-indexed logic (even \LTLmX\ and \CTLmX) is undecidable, and therefore does not have cutoffs. Moreover, if processes can control both in- and out-directions, then the PMCP is already undecidable for bi-directional rings and $1$-indexed \LTLmX.

\blackout{
\subsubsection{Summarizing table}

Table~\ref{tab:summary} summarizes what is now known about prenex indexed-temporal logic on token-passing systems with valueless token. All results are new  except that in the first row of the table
the following was known: 
\begin{itemize}
  \item $4$ is a cutoff for $\twoallCTLstarmX$ on rings~\cite{EN95},
  \item $2k$ is a cutoff for $\kLTLmX$ on rings~\cite{KJB13},
  \item there does not exist a cutoff for $\twoallCTLstarmX$ on arbitrary topologies~\cite{CTTV04},
  \item for every $k$ there exists a cutoff for  $\kLTLmX$ on arbitrary topologies~\cite{CTTV04}.
\end{itemize}

\todo{complete table}

\noindent
\begin{table} 
\resizebox{\linewidth}{!}{
\begin{tabular}{ |c|c|c|}
\hline
Dir I/O &  \dkCTLstarmX\  & \piCTLstarmX\ \\
 \hline    
 \hline                    
N/N &  {cutoffs} exist for arbitrary topologies and $\forall k,d$       & cutoffs exist for certain topologies; not in general for $k=2$ \\
N/C & {undecidable} for certain topologies and $d = 1,k = 9$       & {undecidable}  \\
C/N & {undecidable} for certain topologies and  $d = 1,k = 9$     & {undecidable}  \\
C/C & {undecidable} for bi-rings        & {undecidable}  \\
  \hline  
\end{tabular}
}
 \caption{
Summary of cutoff and undecidability results for prenex indexed temporal logic.
The column `Dir I/O' refers to direction-awareness of incoming and outgoing directions: an `N' means the process can not choose directions, and entry `C' means the process can choose directions. Thus `N/C' means that the process can not choose incoming directions but can choose outgoing directions. In the first row, uni-directional rings and bi-directional rings have a cutoff of $2k$, cliques have a cutoff of $k+1$, stars have a cutoff of $k+2$.
 }
 \label{tab:summary}
 \end{table}

{\em Note on missing cases $\kCTLstarmX$ and $\pidCTLstarmX$.} These are undecidable for N/C, C/N, and C/C.
For N/N $\pidCTLstarmX$ has no cutoff on arbitrary topologies already for $d = 1$ since with $k$ quantifiers and suitable process $P$ one can express that the topology has size at least $k$. 
}

\medskip
\noindent
\emph{Technical contributions relative to previous work.}
\blackout{
\todo{This section needs a complete rewrite. Must answer: What is new? What can be reused? What are the similarities and differences with existing techniques?}

\todo{We should very carefully explain how CTTV + EN fit into our framework. We don't want them to think we didn't understand what they did.}
}
Our main positive result (Theorem~\ref{thm:cutoff_kctld}) generalizes
proof techniques and ideas from previous results~\cite{EN95,CTTV04}. We observe that in
both of these papers the main idea is to abstract a TPS by simulating
the quantified processes exactly and simulating the movement of the
token between these processes. The relevant
information about the movement of the token is this: whether there
is a direct edge, or a path (through unquantified processes) from one
quantified process to another. This abstraction does not work for
$\dCTLstarmX$ and general topologies since the formula can express
branching properties of the token movement. Our main observation is
that the relevant information about the branching-possibilities of the
token can be expressed in $\dCTLstarmX$ over the topology itself. We
develop a composition-like theorem, stating that if two topologies
(with $k$ distinguished vertices) are indistinguishable by
$\dCTLstarmX$ formulas, then the TPSs based on these topologies and an
arbitrary process template $P$ are indistinguishable by $\dCTLstarmX$
(Theorem~\ref{theorem:reduction}). The machinery involves a
generalization of stuttering trace-equivalence~\cite{SKS}, a notion of
$d$-contraction that serves the same purpose as the connection
topologies of Clarke et al.~\cite[Proposition $1$]{CTTV04},  and also
the main simulation idea of Emerson and Namjoshi~\cite[Theorem $2$]{EN95}. 

Our main negative result, undecidability of PMCP for direction-aware systems (Theorem~\ref{thm:DA}), is proven by a reduction from the non-halting problem for $2$-counter machines
(as is typical in this area~\cite{EN95,EsparzaFM99}). Due to the lack of space, full proofs are omitted, and can be found in the full version~\cite{FullVersion}.
\sj{Add a remark that a version with full proofs is available as
  TechReport or on ArXiv?\\techreport URL?}

\section{Definitions and Existing Results}


Let $\IndSet$ denote the set of positive integers. Let $[k]$ for
$k\in\IndSet$ denote the set $\{1,\dots,k\}$. 
The concatenation
of strings $u$ and $w$ is written $uw$ or $u \cdot w$. 

Let $\AP$ denote a countably infinite set of {\em atomic propositions} or {\em atoms}. A \emph{labeled transition system (LTS) over $\AP$} is a tuple
$
(Q,Q_0,\Sigma,\delta,\lambda)
$
where 
 $Q$ is the set of {\em states}, $Q_0 \subseteq Q$ are
the {\em initial states}, $\Sigma$ is the set of {\em transition labels} (also called {\em action labels}), $\delta
\subseteq Q \times \Sigma \times Q$ is the {\em transition relation}, and
$\lambda:Q \to 2^{\AP}$ is the {\em state-labeling} and satisfies that $\lambda(q)$ is finite (for every $q \in Q$).
Transitions $(q,\sigma,q') \in \delta$ may
be written $\trans{q}{q'}{\sigma}$.

A \emph{state-action path} of an LTS $(Q,Q_0,\Sigma,\delta,\lambda)$
 is a finite sequence of the form $q_0 \sigma_0 q_1 \sigma_1 \dots q_n \in (Q \Sigma)^*Q$ or an infinite sequence  $q_0\sigma_0 q_1 \sigma_1 \dots \in (Q  \Sigma)^\omega$ such that $(q_i,\sigma_i,q_{i+1}) \in \delta$ (for all $i$).
A \emph{path} of an LTS is the projection $q_0 q_1 \dots$ of a state-action path onto states $Q$.
An \emph{action-labeled path} of an LTS is the projection $\sigma_0 \sigma_1 \dots$ of a state-action path onto transition labels $\Sigma$.

\subsection{System Model (Direction-Unaware)}
\label{dfn:paramsys}\label{def:process_template}
In this section we define the LTS $P^G$ --- it consists of replicated copies of a process $P$ placed on the vertices of a graph $G$. Transitions in $P^G$ are either internal (in which exactly one process moves) or synchronized (in which one process sends the token to another along an edge of $G$). The token starts with the process that is at the initial vertex of $G$.

Fix a countably infinite set of (local) atomic propositions
$\APproc$ (to be used by the states of the individual processes).

\smallskip\noindent\textbf{Process Template $P$.}
Let $\ActionsInt$ denote a finite non-empty set of {\em internal-transition labels}.
Define $\ActionsProc$ as the disjoint union $\ActionsInt \cup \{\trcv, \tsnd\}$ where $\trcv$ and $\tsnd$ are new symbols.

A {\em process template} $P$ is a LTS
    $(\Locals, \LocalsI, \ActionsProc, \Trans, \LabelFun)$ over $\APproc$ such that:
\begin{enumerate}[label*=\roman*)]
    \item the state set $\Locals$ is finite and can be partitioned into two non-empty sets, say $T \cup N$. States in $T$ are said to {\em have the token}.
    \item The initial state set is $Q_0  = \{\iota_t, \iota_n\}$ for some $\iota_t \in T, \iota_n \in N$.
    \item Every  transition $\trans{q}{q'}{\tsnd}$ satisfies that $q$ has the token and $q'$ does not.
\item Every  transition $\trans{q}{q'}{\trcv}$ satisfies that
  $q'$ has the token and $q$ does not.
\item Every  transition $\trans{q}{q'}{a}$ with $\textsf{a} \in  \ActionsInt$ satisfies that $q$ has the token if and only if $q'$ has the token. 

\item The transition relation $\Trans$ is total in the first
  coordinate: for every $q \in \Locals$ there exists $\sigma \in
  \ActionsProc, q' \in \Locals$ such that $(q,\sigma,q') \in \Trans$
  (i.e., the process $P$ is non-terminating).

\item Every infinite action-labeled path
$a_0 a_1 \dots$ is in the set 
$(\ActionsInt^* \ \tsnd\  \ActionsInt^* \ \trcv)^\omega \cup 
 (\ActionsInt^* \ \trcv \  \ActionsInt^*\  \tsnd)^\omega$ 
(i.e., $\tsnd$ and $\trcv$ actions alternate continually along every
  infinite action-labeled path of $P$).
\footnote{This restriction was introduced by Emerson and Namjoshi in~\cite{EN95}. Our positive results that cutoffs exist also hold for a more liberal restriction (see Section~\ref{sec:extensions}).}
\end{enumerate}
           
The elements of $\Locals$ are called {\em local states} and the
transitions in $\Trans$ are called {\em local transitions (of
  $P$)}. 
A local state $q$ such that the only transitions are of the form $\trans{q}{q'}{\tsnd}$ (for some $q'$) is said to be {\em send-only}. A local state $q$ such that the only transitions are of the form $\trans{q}{q'}{\trcv}$ (for some $q'$) is said to be {\em receive-only}.

\smallskip\noindent\textbf{Topology $G$.}
 A {\em topology}
 is a directed graph $G = (V,E,x)$ where $V = [k]$ for some $k \in \IndSet$, vertex $x \in V$ is the {\em initial vertex}, $E \subseteq V \times V$,  and $(v,v) \not \in E$ for every $v \in V$. Vertices are
 called {\em process indices}.

We may also write $G = (V_G,E_G,x_G)$ if we need to disambiguate.

\smallskip\noindent\textbf{Token-Passing System $P^G$.}
	Let $\APsys := \APproc  \times \IndSet$ be the \emph{indexed atomic propositions}. 
	For $(p,i) \in \APsys$ we may also write $p_i$. 
	Given a process template $P =  (\Locals, \LocalsI, \ActionsProc, \Trans, \LabelFun)$ over $\APproc$ and a topology $G = (V,E,x)$, define the {\em token-passing system (TPS)} $P^G$ as the finite LTS 
$(S,S_0,\ActionsInt \cup \{\tok\},\Delta,\Lambda)$
over atomic propositions $\APsys :=  \APproc \times \IndSet$,
 where:
\begin{itemize}

%

\item The set $S$ of \emph{global states} is $Q^V$, i.e., all functions from $V$ to $Q$. If $s \in Q^V$ is a global state then $s(i)$ denotes the local state of the process with index $i$.

\item 
   The set of \emph{global initial states} $S_0$
    consists of the unique global state $s_0 \in Q_0^{V}$ such  that only
    $s_0(x)$ has the token (here $x$ is the initial vertex of $G$).

\item The labeling $\Lambda(s) \subset \APsys$ for $s \in S$ is defined as follows: $p_i \in \Lambda(s)$ if and only if $p \in \LabelFun(s(i))$, for $p \in \APproc$ and $i \in V$.

\item The {\em global transition relation} $\Delta$ is defined to consist of the set of all internal transitions and synchronous transitions:


%
%
    


\begin{itemize}
\item An {\em internal transition} is an element $(s,{\sf a},s')$ of $S \times \ActionsInt \times S$ for which there exists a process index $v \in V$ such that 
\begin{enumerate}[label*=\roman*)]
\item $\trans{s(v)}{s'(v)}{a}$ is a local transition of $P$, and 
\item for all $w \in V \setminus \{v\}$, $s(w) = s'(w)$.
\end{enumerate}

\item
A {\em token-passing transition}  is an element $(s,\tok,s')$ of $S \times \{\tok\} \times S$ for which there exist process indices $v,w \in V$ such that $(v,w) \in E$ and

\begin{enumerate}[label*=\roman*)]
  \item $\trans{s(v)}{s'(v)}{\tsnd}$ is a local transition of $P$,      
  \item $\trans{s(w)}{s'(w)}{\trcv}$ is a local transition of $P$, and
  \item for every $u \in V \setminus \{v,w\}$, $s'(u) = s(u)$. 
 \end{enumerate}
\end{itemize}

\end{itemize}

In words, the system $P^G$ can be thought of the asynchronous parallel composition of $P$ over topology $G$. The token starts with process $x$. At each time step either exactly one process makes an internal transition, or exactly two processes synchronize when one process sends the token to another along an edge of $G$.



\subsection{System Model (Direction-Aware)}
\label{sec:DATPS}
Inspired by direction-awareness in the work of Emerson and
Kahlon~\cite{EK04}, we extend the definition of TPS  
to include additional labels on edges, called \emph{directions}. 
The idea is that processes can restrict which directions are used when they send or receive the token. \sr{add more intuitive explanation. mention unbounded degree. eg: so when a process sends in direction $d$ the token non-deterministically goes along an edge labeled $d$ (there may be unboundedly many such edges)} \sr{add picture}

Fix finite non-empty disjoint sets $\dirset_{\tsnd}$ of \emph{sending directions} and $\dirset_{\trcv}$ of 
\emph{receiving directions}. A \emph{direction-aware token-passing system} is a TPS with the following modifications.


\smallskip\noindent\textbf{Direction-aware Topology.}\label{def:graph_with_dir}
A {\em direction-aware topology} is a topology $G = (V,E,x)$ 
with labeling functions 
$\dir_\trcv: E \rightarrow \dirset_\trcv$,
$\dir_\tsnd: E \rightarrow \dirset_\tsnd$. 
%



\smallskip\noindent\textbf{Direction-aware Process Template.}
For process templates of direction-aware systems, 
transition labels are taken from $\ActionsProc := \ActionsInt \cup 
\dirset_{\tsnd} \cup 
\dirset_{\trcv}$. 
The definition of a \emph{direction-aware process template} is like
that in Section~\ref{dfn:paramsys},
except that 
in item iii) $\tsnd$ is replaced by ${\sf d} \in \dirset_{\tsnd}$, 
in iv) $\trcv$ is replaced by ${\sf d} \in \dirset_{\trcv}$, and 
in vii) $\tsnd$ is replaced by $\dirset_{\tsnd}$ and $\trcv$ by $\dirset_{\trcv}$
.

\smallskip\noindent\textbf{Direction-aware Token Passing System.}
Fix $ \dirset_{\tsnd}$ and $\dirset_{\trcv}$, 
let $G$ be a direction-aware topology and $P$ a direction-aware process template.
%
Define the \emph{direction-aware token-passing system} $P^G$ as in
Section~\ref{dfn:paramsys}, except that token-passing transitions are now direction-aware:
%
%
\emph{direction-aware token-passing transitions} are elements
$(s,\tok,s')$ of $S \times \{\tok\} \times S$ for which 
there exist process indices $v,w \in V$ with $(v,w) \in E$, 
$\dir_\tsnd(v,w) = {\sf d}$, and 
$\dir_\trcv(v,w) = {\sf e}$, such that:

\begin{enumerate}[label*=\roman*)]

  \item
    $\trans{s(v)}{s'(v)}{d}$ is a local transition of
    $P$.

  \item
    $\trans{s(w)}{s'(w)}{e}$ is a
    local transition of $P$.
                 
  \item For every $u \in V \setminus \{ v,w\}$, $s'(u) = s(u)$. 

\end{enumerate}\sr{i think we should provide an intuition about what this means. eg. if a process wants to send in direction d but no process wants to receive from my direction then the system blocks.}

\smallskip\noindent\textbf{Notations $\Psimptok$, $\Psnd$, $\Prcv$, $\Prcvsnd$.}
Let $\Psimptok$ denote the set of all process templates for which
$\card{\dirset_{\tsnd}} = \card{\dirset_{\trcv}} = 1$. In this case
$P^G$ degenerates to a direction-unaware TPS as defined in 
Section~\ref{dfn:paramsys}.
If we require $\card{\dirset_{\trcv}} = 1$, then processes
cannot choose from which directions to receive the token, 
but possibly in which direction to send it.
Denote the set of all such process templates by $\Psnd$.
Similarly define $\Prcv$ to be all process templates where
$\card{\dirset_{\tsnd}} = 1$ --- processes cannot choose where to send
the token,  
but possibly from which direction to receive it. 
Finally, let $\Prcvsnd$ be the set of all direction-aware process templates.

\smallskip\noindent\textbf{Examples.}
Figure~\ref{fig:biring} shows a bi-directional ring with
directions {\sf cw} (clockwise) and {\sf ccw}
(counterclockwise). Every edge $e$ is labeled with
an outgoing 
direction $dir_\tsnd(e)$ and an incoming direction
$dir_\trcv(e)$.\footnote{For notational simplicity,
  we denote both outgoing direction $\tsnd_{\sf cw}$ and incoming
  direction $\trcv_{\sf cw}$ by {\sf cw}, and similarly for {\sf ccw}.} 
Using these directions, a process that has the token can choose
whether he wants to send it in direction {\sf cw} or {\sf
  ccw}. Depending on its local state, a process waiting for the token
can also choose to receive it only from direction {\sf cw} or {\sf
  ccw}.

Figure~\ref{fig:mult_directions} depicts a topology in which process
$1$ can choose between two outgoing directions. If it sends
the token in direction $\tsnd_1$, it may be received by either process
$2$ or $3$. If however process $2$ blocks receiving from direction
$\trcv_1$, the token can only be received by process $3$. If
$3$ additionally blocks receiving from $\trcv_2$, then this
token-passing transition is disabled.
\begin{figure}[t]
\vspace{-0.2cm}
\centering
\subfigure[Bi-directional ring with directions.]{%
\makebox[0.47\linewidth]{%
\scalebox{0.9}{
\begin{tikzpicture}[->,
node distance = 2cm, 
auto, 
semithick, 
inner sep=.01cm,
bend angle=55]
    \tikzset{every state/.style={rectangle,rounded corners,minimum size=2em}}
    \tikzset{every edge/.append style={font=\small, right}}
    \tikzset{box state/.style={draw,rectangle,rounded corners,inner sep=.1cm}}
    
    \tikzset{SnakeStyle/.style = {snake,segment amplitude=.2mm,segment length=1mm,line after snake=2mm}}
    
    \node[state] (1) {$1$};
    \node[state] (2) [below right of=1] {$2$};
    \node[state] (3) [below left of=2] {$3$};
    \node[state] (4) [below left of=1] {$4$};
    
    \path (1)  edge[bend left] node[anchor=south,near start,above,yshift=2pt] {\sf cw} 
                               node[anchor=west,near end,right,xshift=1pt] {\sf cw} (2);
    \path (2)  edge[bend right=35] node[anchor=north,near end,below,xshift=-3pt,yshift=-3pt] {\sf ccw} 
                               node[anchor=east,near start,left,xshift=-1pt,yshift=-3pt] {\sf ccw} (1);
    \path (2)  edge[bend left] node[anchor=west,near start,right,xshift=1pt] {\sf cw} 
                               node[anchor=north,near end,below,yshift=-2pt] {\sf cw} (3);
    \path (3)  edge[bend right=35] node[anchor=east,near end,left,xshift=-1pt,yshift=3pt] {\sf ccw} 
                               node[anchor=south,near start,above,xshift=-1pt,yshift=3pt] {\sf ccw} (2);
    \path (3)  edge[bend left] node[anchor=north,near start,below,yshift=-2pt] {\sf cw} 
                               node[anchor=east,near end,left,xshift=-1pt] {\sf cw} (4);
    \path (4)  edge[bend right=35] node[anchor=south,near end,above,xshift=1pt,yshift=3pt] {\sf ccw} 
                               node[anchor=west,near start,right,xshift=1pt,yshift=3pt] {\sf ccw} (3);
    \path (4)  edge[bend left] node[anchor=east,near start,left,xshift=-1pt] {\sf cw} 
                               node[anchor=south,near end,above,yshift=2pt] {\sf cw} (1);
    \path (1)  edge[bend right=35] node[anchor=west,near end,right,xshift=1pt,yshift=-3pt] {\sf ccw} 
                               node[anchor=north,near start,below,xshift=3pt,yshift=-3pt] {\sf ccw} (4);
    
    
    
    
    
    
\end{tikzpicture}
}
}
\label{fig:biring}}
\quad
\subfigure[Directed topology with multiple edges of same direction.]{%
\makebox[0.47\linewidth]{%
\scalebox{0.9}{
\begin{tikzpicture}[->,
node distance = 2cm, 
auto, 
semithick, 
inner sep=.01cm,
bend angle=55]
    \tikzset{every state/.style={rectangle,rounded corners,minimum size=2em}}
    \tikzset{every edge/.append style={font=\small, right}}
    \tikzset{box state/.style={draw,rectangle,rounded corners,inner sep=.1cm}}
    
    \tikzset{SnakeStyle/.style = {snake,segment amplitude=.2mm,segment length=1mm,line after snake=2mm}}
    
    \node[state] (1) {$1$};
    \node[state] (2) [above right of=1] {$2$};
    \node[state] (3) [right of=1] {$3$};
    \node[state] (4) [below right of=1] {$4$};
    
    \path (1)  edge node[anchor=east,near start,left,yshift=2pt] {$\tsnd_1$} 
                    node[anchor=east,near end,left,yshift=2pt] {$\trcv_1$} (2);
    \path (1)  edge node[anchor=north,near end,below,yshift=-1pt] {$\trcv_2$} 
                    node[anchor=north,near start,below,yshift=-1pt] {$\tsnd_1$} (3);
    \path (1)  edge node[anchor=east,near end,left,yshift=-2pt] {$\trcv_3$} 
                    node[anchor=east,near start,left,yshift=-2pt] {$\tsnd_2$} (4);
    
    
    
    
    
    
\end{tikzpicture}}
}
\label{fig:mult_directions}}
\caption{Direction-aware Topologies.}
\label{fig:fig_directions}
\vspace{-0.4cm}
\end{figure}
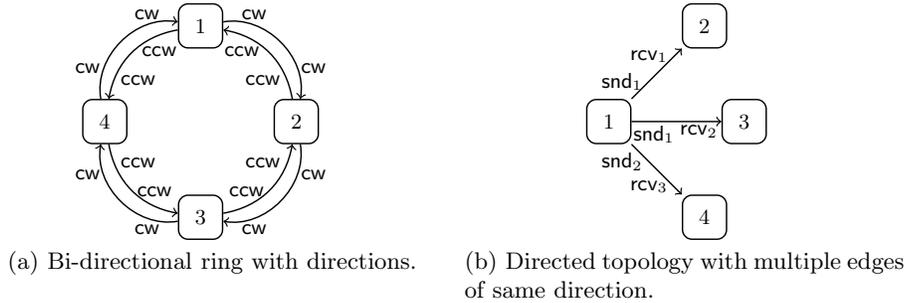
\subsection{Indexed Temporal Logics}
\label{sec:definition-itl}

\emph{Indexed temporal logics (ITL)} were introduced in \cite{BCG1989,ES94,EN95} to model specifications of certain distributed systems. Subsequently one finds a number of variations of indexed temporal-logics
in the literature (linear vs. branching, restrictions on the
quantification). Thus we introduce Indexed-$\CTLstar$ which has these variations (and those in \cite{CTTV04})  as syntactic fragments. 

\sr{add as footnote? ITL formulas are interpreted over systems of the form $P^G$. Although transitions of such systems are labeled by actions, ITL formulas only talk about the atomic propositions/labeling.\\SJ. YES, see survey}




\smallskip\noindent\textbf{Syntactic Fragments of \CTLstar, and $\equiv_\TL$.}
We assume the reader is familiar with the syntax and semantics of \CTLstar, for a reminder see \cite{PrinciplesMC}.
For $d \in \Nat$ let \dCTLstarmX\ denote the syntactic fragment of \CTLstarmX\ in which the nesting-depth of path quantifiers is at most $d$ (for a formal definition see \cite[Section $4$]{SKS}).

Let \TL\ denote a temporal logic (in this paper these are fragments of \CTLstarmX). For temporal logic \TL\ and LTSs $M$ and $N$, write $M \equiv_{\TL} N$ to mean that for every formula $\phi \in \TL$, $M \models \phi$ if and only if $N \models \phi$.

\medskip\noindent\textbf{Indexed-$\CTLstar$.}
%
 Fix an infinite set $\VarNames = \{x,y,z,\dots\}$ of index
 variables, i.e., variables with values from $\IndSet$. These variables refer to vertices in the topology.

%
\smallskip\noindent\textit{Syntax.}\sr{the rest of this section has been reworked. proofread}
The {\em Indexed-$\CTLstar$\ formulas over variable
set $\VarNames$ and atomic propositions $\AP$} are formed by
adding the following rules to the syntax of $\CTLstar$\ over atomic propositions $\ \AP \times \VarNames$.
We write $p_x$ instead of $(p,x) \in \AP \times \VarNames$.

If $\phi$ is an indexed-$\CTLstar$\ state (resp. path) formula and $x,y \in \VarNames, Y \subset \VarNames$, then
the following are also indexed-$\CTLstar$ state (resp. path) formulas:
\begin{itemize}
\item $\forall x.\ \phi$ and $\exists x. \phi$ (i.e., for all/some vertices in the topology, $\phi$ should hold),
\item $\forall x.\ x \in Y \to \phi$ and $\exists x. x \in Y \wedge \phi$ (for all/some vertices that that are designated by variables in $Y$),
\item $\forall x.\ x \in E(y) \to \phi$ and $\exists x. x \in E(y) \wedge \phi$ (for all/some vertices to which there is an edge from the vertex designated by the variable $y$).
\end{itemize}

\smallskip\noindent\textit{Index Quantifiers.}
We use the usual shortands (e.g., $\forall x \in Y.\ \phi$ is shorthand for $\forall x.\ x \in Y \to \phi$).
The quantifiers introduced above are called called {\em index quantifiers}, denoted $Qx$.



\smallskip\noindent\textit{Semantics.}
Indexed temporal logic is interpreted over a system instance $P^G$ (with $P$ a process template and $G$ a topology). The formal  semantics are in the full version of the paper~\cite{FullVersion}.
Here we give some examples. The formula $\forall i. \ctlE \eventually p_i$ states that for every process there exists a path such that, eventually, that process is in a state that satisfies atom $p$. The formula $\ctlE \eventually \forall i. p_i$ states that there is a path such that eventually all processes satisfy atom $p$ simultaneously. We now define the central fragment that includes the former example and not the latter.

\smallskip\noindent\textbf{Prenex indexed \TL\ and \kTL.}
 \emph{Prenex indexed temporal-logic} is a syntactic fragment of indexed temporal-logic in which all quantifiers 
are at the front of the formula, e.g., prenex indexed $\LTLmX$\ consists of formulas of the form $(Q_1x_1) \dots (Q_kx_k)\ \varphi$ where $\varphi$ is an $\LTLmX$\ formula over atoms $\AP \times \{x_1,\dots,x_k\}$, and the $Q_ix_i$s are index quantifiers. Such formulas with $k$ quantifiers will be referred to as {\em $k$-indexed}, collectively written $\kTL$. The union of $\kTL$ for $k \in \Nat$ is written $\piTL$ and called (full) prenex indexed \TL. 

\todo{add note that there is no prenex normal form}

\subsection{Parameterized Model Checking Problem, Cutoffs, Decidability}
\label{sec:definition-pmcp}\label{sec:definition-cutoff}
A {\em parameterized topology} $\pgraph$ is a countable set of topologies. E.g., the set of uni-directional rings with all possible initial vertices is a parameterized topology.\sr{ends abruptly}

\smallskip
\noindent
{\bf $\PMCP_{\pgraph}(-,-)$.}
The {\em parameterized model checking problem (PMCP)} for parameterized topology $\pgraph$, processes from $\Pproctemp$, and parameterized specifications from $\Pspec$, written $\PMCP_{\pgraph}(\Pproctemp,\Pspec)$, is the set of pairs $(\pspec,P) \in \Pspec \times \Pproctemp$ such that for all $G \in \pgraph$, $P^G \models \pspec$. 
A {\em solution} to $\PMCP_{\pgraph}(\Pproctemp,\Pspec)$ is an algorithm that, given a formula $\pspec \in \Pspec$ and a process template $P \in \Pproctemp$ as input, outputs 'Yes' if for all $G \in \pgraph$, $P^G \models \pspec$, and 'No' otherwise.

%
%

\medskip
\noindent
{\bf Cutoff.}
A {\em cutoff} for $\PMCP_\pgraph(\Pproctemp,\Pspec)$ is a natural number $c$ such that for every $P \in \Pproctemp$ and $\pspec \in \Pspec$, the following are equivalent:
\begin{itemize}
\item $P^G \models \pspec$ for all $G \in \pgraph$ with $|V_G| \leq c$;
\item $P^G \models \pspec$ for all $G \in \pgraph$.
\end{itemize}

Thus $\PMCP_\pgraph(\Pproctemp,\Pspec)$ {\em does not have a cutoff} iff for every $c \in \Nat$ there exists $P \in \Pproctemp$ and $\pspec \in \Pspec$ such that $P^G \models \pspec$ for all $G \in \pgraph$ with $|V_G| \leq c$, and there exists $G \in \pgraph$ such that $P^G \not \models \pspec$.\sr{compare to EN + CTTV cutoff}

\begin{observation}
\label{obs:cutoffs_implies_decidability}
If $\PMCP_\pgraph(\Pproctemp,\Pspec)$ has a cutoff, then
$\PMCP_\pgraph(\Pproctemp,\Pspec)$ is decidable

\end{observation}

{Indeed: if $c$ is a cutoff, let $G_1, \dots, G_n$ be all topologies $G$ in $\pgraph$ such that $|V_G| \leq c$.
The algorithm that solves PMCP takes $P,\varphi$ as input and checks whether or not $P^{G_i} \models \pspec$ for all $1 \leq i \leq n$.}

\medskip
\noindent
{\bf Note About Decidability.} \label{sec:nonuniform}
The following statements are not, a priori, equivalent (for given parameterized topology $\pgraph$ and process templates $\Pproctemp$):
\begin{itemize}
\item[-] For every $k \in \Nat$, $\PMCP_\pgraph(\Pproctemp,\kTL)$ is decidable.
\item[-] $\PMCP_\pgraph(\Pproctemp,\piTL)$ is decidable.
\end{itemize}

The first item says that for every $k$ there {\em exists} an algorithm
$A_k$ that solves the PMCP for $k$-indexed \TL. This does not imply
the second item, which says that there exists an algorithm that
solves the PMCP for $\cup_{k \in \Nat} \kTL$. If the function $k
\mapsto A_k$ is also {\em computable} (e.g.,
Theorem~\ref{thm:explicit_cutoffs}) then indeed the second item
follows: given $P,\varphi$, extract the size $k$ of the prenex block of $\varphi$, {\em compute} a description of $A_k$, and run $A_k$ on $P,\varphi$.

For instance, the result of Clarke et al. --- that there are cutoffs for $k$-index \LTLmX\ and arbitrary topologies --- does not imply that the PMCP for $\piLTLmX$ and arbitrary topologies is decidable. Aware of this fact, the authors state (after Theorem $4$) ``Note that the theorem does not provide us with an effective means to find the reduction [i.e. algorithm]...". 

In fact, we prove (Theorem~\ref{thm:undecprenex}) that there is some parameterized topology such that PMCP is undecidable for prenex indexed \LTLmX.
%


\medskip
\noindent
{\bf Existing Results.} We restate the known results using our terminology.

\noindent
A {\em uni-directional ring} $G = (V,E,x)$ is a topology with $V = [n]$ for some $n \in \Nat$, there are edges $(i,i+1)$   for $1 \leq i \leq n$ (arithmetic is modulo $n$), and $x \in V$.
Let $\pring$ be the parameterized topology consisting of all uni-directional rings.

\begin{theorem}[Implict in \cite{EN95}]
For every $k \in \Nat$, there is a cutoff for the problem
$
PMCP_\pring(\Psimptok,\kCTLstarmX).
$
\footnote{The paper explicitly contains the result that $4$ is a cutoff for $\twoallCTLstarmX$ on rings. However the proof ideas apply to get the stated theorem.}
\end{theorem}

Although Clarke et al.~\cite{CTTV04} do not explicitly state the
following theorem, it follows from their proof technique, which we
generalize in Section~\ref{sec:method}. 
\begin{theorem}[Implicit in \cite{CTTV04}]\label{thm:cttv_reworded}
For every parameterized topology $\pgraph$, and every $k \in \Nat$, the problem 
$\PMCP_\pgraph(\Psimptok,\kLTLmX)$ has a cutoff. 
\footnote{Khalimov et al.~\cite[Corollary 2]{KJB13} state that
$2k$ is a cutoff if $\pgraph$ is taken to be $\pring$. However this is an error:
$2k$ is a cutoff only for formulas with no quantifier alternations.
See Remark~\ref{rem:kjb_error} in Section~\ref{sec:explicit_cutoffs}.}.
\end{theorem}
\begin{theorem}[{\cite[Corollary $3$]{CTTV04}}]
\label{thm:CTTVnocutoff}
There exists a parameterized topology $\pgraph$ and process $P \in \Psimptok$ such that the problem $\PMCP_\pgraph(\{P\},\twoallCTLmX)$ does not have a cutoff.\sr{differing notation, EE}
\end{theorem}

The proof of this theorem defines $\pgraph$, process $P$, and for
every $c \in \Nat$ a formula $\varphi_c$, such that if  $G \in
\pgraph$ then $P^G \models \varphi_c$  if and only if $|V_G| \leq
c$. The formula $\varphi_c$ is in $2$-indexed \CTLmX\ and has nesting
depth of path quantifiers equal to $c$.\todo{inconsistent with remark
  in intro}

\sr{Ends abruptly... Also, shall we state the direction-aware result of EK?}

\section{Method for Proving Existence of Cutoffs} \label{sec:method}
\sr{this section reworked according to swen's observation that initial states can be fixed. proofread}
We give a method for proving cutoffs 
for direction-\emph{un}aware TPSs that will be used to prove
Theorem~\ref{thm:cutoff_kctld}.\sj{is the reference to the theorem
  useful? the reader does not know what it is about.} 

\todo{say something about limitations of method...?}

In a $k$-indexed \TL\ formula $Q_1 x_1 \dots Q_k x_k.\ \varphi$, every valuation of the variables $x_1, \dots, x_k$ designates $k$ nodes of the underlying topology $G$, say $\bar{g} = g_1, \dots, g_k$. The formula $\varphi$ can only talk about (the processes in) $\bar{g}$. In order to prove that the PMCP has a cutoff, it is sufficient 
(as the proof of Theorem~\ref{thm:compfinicutoff} will demonstrate) to
find conditions on two topologies $G,G'$ and $\bar{g},\bar{g}'$ that
allow one to conclude that $P^G$ and $P^{G'}$ are indistinguishable
with respect to $\varphi$.\sj{with respect to every $\varphi$ in the
  logic, rather?}

We define two abstractions for a given TPS $P^G$. The first
abstraction simulates $P^G$, keeping track only of the local states of
processes indexed by $\bar{g}$. We call it the \emph{projection} of
$P^G$ onto $\bar{g}$.\footnote{Emerson and Namjoshi~\cite[Section $2.1$]{EN95}
define the related notion ``LTS projection".} The second abstraction
only simulates the movement of the token in $G$, restricted to
$\bar{g}$. We call it the \emph{graph LTS} of $G$ and $\bar{g}$.

{\em Notation.} Let $\bar{g}$ denote a tuple $(g_1,\dots,g_k)$ of {\em distinct} elements of $V_G$, and $\bar{g}'$ a $k$-tuple of distinct elements of $V_{G'}$. Write $v \in \bar{g}$ if $v = g_i$ for some $i$.

\medskip
\noindent
{\bf The projection $\abstraction{P^G}{\bar{g}}$.}
Informally, the {\em projection} of $P^G$ onto a tuple of process
indices $\bar{g}$ is the LTS $P^G$ and a new labeling\sj{``with a new
  labeling''? or ``with a modified labeling''?} that, for every $g_i \in \bar{g}$, replaces the indexed atom $p_{g_i}$ by the atom $p@i$; all other atoms are removed. Thus $p@i$ means that the atom $p \in \APproc$ holds in the process with index $g_i$. In other words, process indices are replaced by their {\em positions} in $\bar{g}$.


Formally, fix process $P$, topology $G$, and $k$-tuple $\bar{g}$ over $V_G$. Define the {\em projection of $P^G=(S,S_0,\ActionsInt \cup \{\tok\},\Delta,\Lambda)$ onto $\bar{g}$}, written $\abstraction{P^G}{\bar{g}}$
as the LTS $(S,S_0,\ActionsInt \cup \{\tok\},\Delta,L)$
over atomic propositions $\{p@i : p \in \APproc, i \in [k]\}$, where for all $s \in S$ the labeling $L(s)$ is defined as follows: $L(s) := \{p@i : p_{g_i} \in \Lambda(s), i \in [k]\}$.

\medskip
\noindent
{\bf The graph LTS $\tokengraph{G}{\bar{g}}$.}
Informally, $\tokengraph{G}{\bar{g}}$ is an LTS where states are nodes of the graph $G$, and transitions are edges of $G$. The restriction to $\bar{g}$ is modeled by labeling a state with the position of the corresponding node in $\bar{g}$.
\sr{add footnote citing nominals.}

Let $G = (V,E,x)$ be a topology, and let $\bar{g}$ be a $k$-tuple over $V_G$.
Define the {\em graph LTS $\tokengraph{G}{g}$} as the LTS
$(Q,Q_0,\Sigma,\Delta,\Lambda)$ over atomic propositions
$\{1,\dots,k\}$, with state set $Q :=V$, initial state set $Q_0 :=
\{x\}$, action set $\Sigma = \{{\sf a}\}$, transition relation with
$(v,{\sf a},w) \in \Delta$ iff $(v,w) \in E$, and labeling $\Lambda(v) := \{i\}$ if $v = g_i$ for some $1 \leq i \leq k$, and $\emptyset$ otherwise.\footnote{Atomic propositions that have to be true in exactly one state of a structure are called nominals in \cite{BLMV08,SaVa01}.}

\blackout{
$\lambda(v) := 
\begin{cases}

\{i\} 		& \mbox{if } v = g_i \mbox{ for some } 1 \leq i \leq k\\
\emptyset & \mbox{otherwise.}

\end{cases}
$
}
%
%




\medskip
\noindent
Fix a non-indexed temporal logic \TL, such as $\CTLstarmX$. We now define what it means for \TL\ to have the reduction property and the finiteness property. Informally, 
the reduction property says that if $G$ and $G'$ have the same connectivity (with respect to \TL\ and only viewing $k$-tuples $\bar{g},\bar{g}'$) then $P^G$ and $P^{G'}$ are indistinguishable (with respect to \TL-formulas over process indices in $\bar{g},\bar{g}'$).

\noindent
{\bf {\scshape reduction} property for \TL}\footnote{Properties of this type are sometimes named {\em composition} instead of {\em reduction}, see for instance \cite{rabi07}.}

\smallskip

\begin{tabular}{l}
For every $k \in \Nat$,  process $P \in \Psimptok$, topologies $G,G'$,
$k$-tuples ${\bar{g}},{\bar{g}'}$,\\
if $\tokengraph{G}{\bar{g}} \equiv_{\TL} \tokengraph{G'}{\bar{g}'}$ then 
$\abstraction{P^{G}}{\bar{g}} \equiv_{\TL} \abstraction{P^{G'}}{\bar{g}'}.$
\end{tabular}

\medskip
\noindent
{\bf {\scshape finiteness}  property for \TL}

\smallskip
\begin{tabular}{l}
For every $k \in \Nat$, there are finitely many equivalence classes
$[\tokengraph{G}{\bar{g}}]_{\equiv_{\TL}}$\\ where $G$ is an arbitrary topology, and ${\bar{g}}$ is a $k$-tuple over $V_G$.
\end{tabular}

\begin{theorem}[\textsc{reduction} \& \textsc{finiteness} $\implies$ Cutoffs for $\kTL$]
\label{thm:compfinicutoff}
If $\TL$ satisfies the \textsc{reduction} and the \textsc{finiteness} property, then
for every $k,\pgraph$,
$
\PMCP_{\pgraph}(\Psimptok,\kTL)
$
has a cutoff.
\end{theorem}

%

\begin{proof}

Fix quantifier prefix $Q_1 x_1 \dots Q_k x_k$.  We prove that there
exist finitely many topologies $G_1, \cdots, G_N \in \pgraph$ such
that for every $G \in \pgraph$ there is an $i \leq N$ such that for
all $P \in \Psimptok$, and all $\TL$-formulas $\varphi$ over atoms
$\APproc \times \{x_1,\cdots,x_k\}$ 
\[
P^G \models Q_1 x_1 \dots Q_k x_k.\ \varphi \iff P^{G_i} \models Q_1 x_1 \dots Q_k x_k.\ \varphi
\]
In particular, $\max\{|V_{G_i}| : 1 \leq i \leq N\}$ is a cutoff for $\PMCP_{\pgraph}(\Psimptok,\kTL)$.

Suppose
for simplicity of exposition that $Q_i x_i$ is a quantifier that also expresses that the value of $x_i$ is different from the values of $x_j \in \{x_1, \dots, x_{i-1}\}$.\footnote{All the types of quantifiers defined in Section~\ref{sec:definition-itl}, such as $\exists x \in E(y)$, can be dealt with similarly at the cost of notational overhead.} 
Fix representatives of $[\tokengraph{G}{\bar{g}}]_{\equiv_{\TL}}$ and define a function $r$ that maps $\tokengraph{G}{\bar{g}}$ to the representative of $[\tokengraph{G}{\bar{g}}]_{\equiv_{\TL}}$. 
Define a function $rep$ that maps $\abstraction{P^{G}}{\bar{g}}$ to $\abstraction{P^{H}}{\bar{h}}$, where $r (\tokengraph{G}{\bar{g}})= \tokengraph{H}{\bar{h}}$. 

For every $\equiv_\TL$-representative $\tokengraph{H}{\bar{h}}$ (i.e., $\tokengraph{H}{\bar{h}} = r(\tokengraph{G}{\bar{g}})$ for some topology $G$ and $k$-tuple $\bar{g}$), introduce a new Boolean proposition 
$q_{\tokengraph{H}{\bar{h}}}$. By the \textsc{finiteness} property of \TL\ there are finitely many such Boolean propositions, say $n$. 

 Define a valuation 
$e_\varphi$ (that depends on $\varphi$) of these new atoms by
\[
e_\varphi(q_{\tokengraph{H}{\bar{h}}}) :=
\begin{cases}
\top &\mbox{if } \abstraction{P^{H}}{\bar{h}}\models \varphi[p_{x_j} \mapsto p@j]\\
\bot &\mbox{otherwise}.
\end{cases}
\]

For every $G \in \pgraph$ define Boolean formula
$
B_{G} := (\overline{Q_1} g_1 \in V_G) \dots (\overline{Q_k} g_k \in V_G) \ q_{r(\tokengraph{G}{\bar{g}})},
$
where $\overline{Q}$ is the Boolean operation corresponding to $Q$,
e.g., $\overline{\exists} g_i \in V_G$ is interpreted as $\bigvee_{g
  \in V_G \setminus \{g_1,\dots,g_{i-1}\}}$.\footnote{Note that in the Boolean propositions $G$ is fixed while $\bar{g}=(g_1,\dots,g_n)$ ranges over $(V_G)^k$ and is determined by the quantification.}

Then (for all $P,G$ and $\varphi$)
\begin{align}
 & P^G \models Q_1 x_1 \dots Q_k x_k.\ \varphi \nonumber \\
&\iff	Q_1 g_1 \in V_G \dots Q_k g_k \in V_G: P^G \models \varphi[p_{x_j} \mapsto p_{g_j}] \nonumber\\
&\iff	Q_1 g_1 \in V_G \dots Q_k g_k \in V_G:  \abstraction{P^{G}}{\bar{g}} \models \varphi[p_{x_j} \mapsto p@j] \nonumber\\
&\iff Q_1 g_1 \in V_G \dots Q_k g_k \in V_G:  rep(\abstraction{P^{G}}{\bar{g}}) \models \varphi[p_{x_j} \mapsto p@j]\nonumber\\
&\iff e_\varphi(B_{G}) = \top \nonumber
\end{align}
Here $\varphi[p_{x_j} \mapsto p_{g_j}]$ is the formula resulting from replacing every atom in $\varphi$ of the form $p_{x_j}$ by the atom $p_{g_j}$, for $p \in \APproc$ and $1 \leq j \leq k$. Similarly $\varphi[p_{x_j} \mapsto p@j]$ is defined as the formula resulting from replacing (for all $p \in \APproc, j \leq k$) every atom in $\varphi$ of the form $p_{x_j}$  by the atom $p@j$. The first equivalence is by the definition of semantics of indexed temporal logic; the second is by the definition of $\abstraction{P^{G}}{\bar{g}}$; the third is by the \textsc{reduction} property of \TL; the fourth is by the definition of $e_\varphi$ and $rep$.



Fix $B_{G_1}, \dots, B_{G_N}$ (with $G_i \in \pgraph$) such that every $B_{G}$ ($G \in \pgraph$) is logically equivalent to some $B_{G_i}$. Such a finite set of formulas exists since there are $2^{2^n}$ Boolean formulas (up to logical equivalence) over $n$ Boolean propositions, and thus at most $2^{2^n}$ amongst the $B_G$ for $G \in \pgraph$.

%

By the equivalences above conclude  that for every $G \in \pgraph$ there exists $i \leq N$ such that $P^G \models Q_1 x_1 \dots Q_k x_k.\ \varphi$ if and only if $ P^{G_i} \models Q_1 x_1 \dots Q_k x_k.\ \varphi$. Thus `$\forall G \in \pgraph, P^G \models \varphi$' is equivalent to `$\bigwedge_{i \leq N} e_{\varphi}(B_{G_i})$' and so the integer $c := \max\{|V_{G_i}| : 1 \leq i \leq N\}$ is a cutoff for $\PMCP_{\pgraph}(\Psimptok,\kTL)$.
\qed
\end{proof}

\blackout{
At this point we can already conclude that `$\forall G \in \pgraph, P^G \models \varphi$' is equivalent to `$\bigwedge_{i \leq n} e_{\varphi}(B_{G_i,x_i})$', and thus we have reduced the PMCP to finitely many model-checking problems. Moreover the integer $c := \max\{|V_{G_i}| : 1 \leq i \leq n\}$ is a cutoff for $\PMCP_{\pgraph}(\Psimptok,\kTL)$.
}

%


\blackout{
For every integer $k \in \Nat$, the problem $\PMCP_{\pgraph}(\Pproctemp,\kTL)$ is decidable because one just needs to decide the truth value of `$\bigwedge_{i \leq n} e_{\varphi}(B_{G_i})$' (which can be done since processes in $\Pproctemp$ are finite-state). 
 However, the proof of Theorem~\ref{thm:compfinicutoff} does not show, for general topologies, how to {\em compute} a cutoff from $k$ (the identical situation occurs in our rewording of a result from Clarke et al., Theorem~\ref{thm:cttv_reworded}). Thus the proof does not show that $\PMCP_{\pgraph}(\Pproctemp,\piTL)$ is decidable.  On the other hand, for certain topologies $\pgraph$ and temporal logics \TL, our result is constructive and we can compute cutoffs given $k$ (Theorem~\ref{thm:explicit_cutoffs}). In these cases $\PMCP_{\pgraph}(\Pproctemp,\piTL)$ is decidable: the algorithm extracts $k$ from the input indexed-\TL\ formula, computes a description of the algorithm computing a cutoff for $k$-index \TL\ and runs the algorithm above.
}

\begin{remark} \label{rem:starTL}
The theorem implies that for every $k,\pgraph$, $\PMCP_\pgraph(\Pproctemp,\kTL)$ is decidable. Further, fix $\pgraph$ and suppose that given $k$ one could compute the finite set $G_1, \cdots, G_N$. Then by the last sentence in the proof one can compute the cutoff $c$. In this case, $\PMCP_\pgraph(\Pproctemp,\piTL)$ is decidable.
\end{remark}

\note{Does cutoff + reduction imply finiteness?}


\section{Existence of Cutoffs for $k$-indexed \dCTLstarmX }\label{sec:cutoffsdk}
The following theorem answers Question $1$ from the introduction.
 
\begin{theorem}\label{thm:cutoff_kctld}
Let $\pgraph$ be a parameterized topology. Then for all $k,d \in \Nat$,  the problem $\PMCP_\pgraph(\Psimptok,\dkCTLstarmX)$ has a cutoff. 
\end{theorem}
%
\begin{corollary}
Let $\pgraph$ be a parameterized topology. Then for all $k,d \in \Nat$, the problem $\PMCP_\pgraph(\Psimptok,\dkCTLstarmX)$ is decidable.
\end{corollary}
To prove the Theorem it is enough, by Theorem~\ref{thm:compfinicutoff}, to show that the logic \dkCTLstarmX\ has the \textsc{reduction} property and the \textsc{finiteness} property.

%
%

\begin{theorem}[Reduction]\label{theorem:reduction}
For all $d,k \in \Nat$, topologies $G,G'$, processes $P \in \Psimptok$, 
 $k$-tuples ${\bar{g}}$ over $V_G$ and  $k$-tuples ${\bar{g}'}$ over $V_{G'}$:

If $\tokengraph{G}{\bar{g}} \equiv_{\dCTLstarmX} \tokengraph{G'}{\bar{g}'}$ then 
$\abstraction{P^{G}}{\bar{g}} \equiv_{\dCTLstarmX} \abstraction{P^{G'}}{\bar{g}'}.
$

\end{theorem}

The idea behind the proof is to show that paths in $P^G$ can be simulated by paths in $P^{G'}$ (and vice versa). Given a path $\pi$ in $P^G$, first project it onto $G$ to get a path $\rho$ that records the movement of the token, then take an equivalent path $\rho'$ in $G'$ which exists since $\tokengraph{G}{\bar{g}} \equiv_{\dCTLstarmX}  \tokengraph{G'}{\bar{g}'}$, and then lift $\rho'$ up to get a path $\pi'$ in $P^{G'}$ that is equivalent to $\pi$. This lifting step uses the assumption that process $P$ is in $\Psimptok$, i.e., $P$ cannot control where it sends the token, or from where it receives it. The proof can be found in the full version of the paper~\cite{FullVersion}.

\todo{explain notation similarities and differences. Eg. where is the $h$ of EN?}

\begin{remark} \label{rem:reduction}
As immediate corollaries we get that the \textsc{reduction} property holds with $\TL = \LTLmX$ (take $d = 1$),  $\CTLstarmX$ (since if the assumption holds with $\TL = \CTLstarmX$ then the conclusion holds with $\TL =  \dCTLstarmX$ for all $d \in \Nat$, and thus also for $\TL = \CTLstarmX$) and, if $P$ is finite, also for $\TL = \CTLmX$ (since \CTLmX\ and \CTLstarmX\ agree on finite structures).\sr{reference to this last fact about finite structures}
\end{remark}


\medskip
\noindent
{\bf Finiteness Theorem.}
%
%
Theorem~\ref{thm:CTTVnocutoff} (\cite[Corollary $3$]{CTTV04}) states that there exists $\pgraph$ such that the problem $\PMCP_\pgraph(\Psimptok,\twoexistsCTLstarmX)$ does not have a cutoff. We observed that the formulas from their result have unbounded nesting depth of path quantifiers. This leads to the idea of stratifying $\CTLstarmX$ by nesting depth.

Recall from Section~\ref{sec:definition-itl} that i) $\dCTLstarmX$
denotes the syntactic fragment of \CTLstarmX\ in which formulas have
path-quantifier nesting depth at most $d$; ii) $M \equiv_{\dCTLstarmX} N$ iff $M$ and $N$ agree on all $\dCTLstarmX$ formulas. Write $[M]_{\dCTLstarmX}$ for the set of all LTSs $N$ such that $M \equiv_{\dCTLstarmX} N$.
\sr{Add dfn to spec section}

Following the method of Section~\ref{sec:method}
we prove that the following \textsc{finiteness} property holds (where
$k$ represents the number of process-index quantifiers in the prenex
indexed temporal logic formula).

\begin{remark}\label{rem:globalfairness}
For ease of exposition we sketch a proof under the assumption that path quantifiers in formulas ignore
  runs in which the token does not visit every process infinitely
  often. This is an explicit restriction in \cite{CTTV04} and implicit in \cite{EN95}.  In the full version~\cite{FullVersion} we remove
  this restriction. For the purpose of this paper this restriction only affects the explicit cutoffs in Theorem~\ref{thm:explicit_cutoffs}.
  \sj{should we include (or hint at) the amended construction? YES but no time.}
\end{remark}

\todo{better notation for contraction? $\hat{G}_{d,x,\bar{g}}$ or $\textsf{con}_d{(G_x)|\bar{g}}$}

\begin{theorem}[Finiteness] \label{thm:finiteness}
For all positive integers $k$ and $d$, there are finitely many equivalence classes $[\tokengraph{G}{\bar{g}}]_{\equiv_{\dCTLstarmX}}$ where $G$ is an arbitrary topology, and ${\bar{g}}$ is a $k$-tuple over $V_G$.
\end{theorem}

\paragraph{Proof Idea.} 
We provide an algorithm that given positive integers $k,d$, topology $G$, $k$-tuple $\bar{g}$ over $V_G$, returns a LTS $ \dcon{G}{\bar{g}}$ such that
$\tokengraph{G}{\bar{g}} \equiv_{\dCTLstarmX}
\dcon{G}{\bar{g}}$. Moreover, we prove that for fixed $k,d$ the range of  $\dcon{\_}{\_}$\ak{if con is a function then add parenthesis around args?} is finite.

Recursively define a marking function $\mu_d$ that associates with each $v \in V_G$ a finite set (of finite strings over alphabet $\mu_{d-1}(V_G)$). For the base case define $\mu_0(v) := \Lambda(v)$, the labeling of $\tokengraph{G}{\bar{g}}$. The marking $\mu_d(v)$ stores (representatives) of all strings of $\mu_{d-1}$-labels of paths that start in $v$ and reach some element in $\bar{g}$. The idea is that $\mu_d(v)$ determines the set of \dCTLstarmX\ formulas that hold in $\tokengraph{G}{\bar{g}}$ with initial vertex $v$, as follows: stitch together these strings, using elements of $\bar{g}$ as stitching points, to get the \dCTLstarmX\ types of the infinite paths starting in $v$. This allows us to define a topology, called the $d$-contraction $ \dcon{G}{\bar{g}}$, whose vertices are the $\mu_d$-markings of vertices in $G$. In the full version~\cite{FullVersion} we prove that $\tokengraph{G}{\bar{g}}$ is $\dCTLstarmX$-equivalent to its $d$-contraction, and that the number of different $d$-contractions is finite, and depends on $k$ and $d$.
\todo{give bound on number of contractions. tower of exponentials}

\medskip
\noindent
{\bf Definition of $d$-contraction $\dcon{G}{\bar{g}}$.} 
Next we define $d$-contractions. 


\todo{Currently written to work under fairness assumption in which ignore paths in which token doesn't visit everyone infinitely often}

\medskip
\noindent
{\em Marking $\mu_d$.}
Fix $k,d \in \Nat$, topology $G$, and $k$-tuple ${\bar{g}}$ over $V_G$. Let $\Lambda$ be the labeling-function of $\tokengraph{G}{\bar{g}}$, i.e., $\Lambda(v) = \{i\}$ if $v = g_i$, and $\Lambda(v) = \emptyset$ for $v \not\in \bar{g}$. For every vertex $v \in V_G$ define a set $X(v)$ of paths of $G$ as follows: a path $\pi = \pi_1 \dots \pi_t$, say of length $t$, is in $X(v)$ if $\pi$ starts in $v$, none of $\pi_1, \dots, \pi_{t-1}$ is in $\bar{g}$, and $\pi_t \in \bar{g}$. Note that $X(g_i) = \{g_i\}$.

Define the marking $\mu_d$ inductively:
\[
\mu_d(v) := 
\begin{cases}
 \Lambda(v) &\text{if } d = 0\\
 \{\destut(\mu_{d-1}(\pi_1)\dots  \mu_{d-1}(\pi_t)) : \pi_1 \dots \pi_t \in X(v), t \in \Nat\} &\text{if } d > 0,
\end{cases}
\]
where $\destut(w)$ is the maximal substring $s$ of $w$ such that for
every two consecutive letters $s_i$ and $s_{i+1}$ we have that $s_i
\neq s_{i+1}$. Informally, remove identical consecutive letters of $w$
to get the `destuttering' $\destut(w)$.

\note{give a better definition of destutter}

\note{make use of this notion in proof of reduction} 

The elements of $\mu_d(v)$ ($d > 0$) are finite strings over the alphabet $\mu_{d-1}(V_G)$. For instance, strings in $\mu_1(v)$ are over the alphabet $\{ \{1\}, \{2\}, \dots, \{k\}, \emptyset\}$.\sr{should we change $\emptyset$ to something like $\{H\}$?}


%

\medskip
\noindent
{\em Equivalence relation $\sim_d$.}
Vertices $v,u \in V_G$ are \emph{$d$-equivalent}, written $u \sim_d
v$, if $\mu_d(v) = \mu_d(u)$. We say that $\sim_d$ \emph{refines} $\sim_j$ if
$u \sim_d v$ implies $u \sim_j v$. 

\begin{lemma}\label{lem:refines}
If $0 \leq j < d$, then $\sim_d$ refines $\sim_j$.
\end{lemma}
\noindent
Indeed, observe that for all nodes $v$, all strings in $\mu_d(v)$ start with the letter $\mu_{d-1}(v)$. Thus $\mu_d(v) = \mu_d(u)$ implies that $\mu_{d-1}(v) = \mu_{d-1}(u)$. In other words, if $u \sim_d v$ then $u \sim_{d-1} v$, and thus also $u \sim_j v$ for $0 \leq j < d$. 

\medskip
\noindent
{\em $d$-contraction $\dcon{G}{\bar{g}}$.} 
Define an LTS $ \dcon{G}{\bar{g}}$ called the $d$-contraction of $\tokengraph{G}{\bar{g}}$ as follows.
The nodes of the contraction are the $\sim_d$-equivalence classes. Put an edge between $[u]_{\sim_d}$ and $[v]_{\sim_d}$ if there exists $u' \in [u]_{\sim_d}, v' \in [v]_{\sim_d}$ and an edge in $G$ from $u'$ to $v'$. The initial state is $[x]_{\sim_d}$ where $x$ is the initial vertex of $G$. The label of $[u]_{\sim_d}$ is defined to be  $\Lambda(u)$ --- this is well-defined because, by Lemma~\ref{lem:refines}, $\sim_d$ refines $\sim_0$.

In the full version~\cite{FullVersion} we prove that $\tokengraph{G}{\bar{g}}$ is $\dCTLstarmX$-equivalent to its $d$-contraction, and that the number of different $d$-contractions is finite, and depends on $k$ and $d$. \qed


\section{Cutoffs for $k$-index \CTLstarmX\ and Concrete Topologies}
\label{sec:explicit_cutoffs}

\todo{give time complexity of decision procedure?}

\todo{define uni-ring, bi-ring, clique, star-topology}

The following two theorems answer Question $2$ from the introduction,
regarding the PMCP for specifications from \piCTLstarmX.

First, the PMCP is undecidable for certain (pathological) parameterized topologies $\pgraph$ and specifications from \piCTLstarmX.
\begin{theorem} \label{thm:undecprenex}
There exists a process $P \in \Psimptok$, and parameterized topologies \pgraph, \pgraphH, such that the following PMCPs are undecidable 
\begin{enumerate}
\item $\PMCP_\pgraph(\{P\},\piLTLmX )$.
\item $\PMCP_\pgraphH(\{P\},\twoCTLmX )$.
\end{enumerate}
Moreover, $\pgraph$ and $\pgraphH$ can be chosen to be computable sets of topologies.
\end{theorem}\sj{add proof idea?\\yes ...}
\sr{also, exact quantification shape is known and should be stated}

Second, PMCP is decidable for certain (regular) parameterized
topologies and specifications from \universalpiCTLstarmX. This
generalizes results from Emerson and Namjoshi \cite{EN95} who show
this result for \allkCTLstarmX\ with $k = 1,2$ and uni-directional ring topologies. By Remark~\ref{rem:globalfairness}, these cutoffs apply under the assumption that we ignore runs that do not visit every process infinitely often.

\begin{theorem}
\label{thm:explicit_cutoffs}
If $\pgraph$ is as stated, then $\PMCP_\pgraph(\Psimptok, \allkCTLstarmX)$ has the stated cutoff.
\begin{enumerate}
\item If $\pgraph$ is the set of uni-directional rings, then $2k$ is a cutoff.
\item If $\pgraph$ is the set of bi-directional rings, then  $2k$ is a cutoff.
\item If $\pgraph$ is the set of cliques, then $k+1$ is a cutoff.
\item If $\pgraph$ is the set of stars, then $k+1$ is a cutoff.
\end{enumerate}
Consequently, for each $\pgraph$ listed, $\PMCP_\pgraph(\Psimptok,\universalpiCTLstarmX)$ is decidable.
\end{theorem}

This theorem is proved following Remark~\ref{rem:starTL}: given $k,d$, we compute a set $G_1,\dots,G_N \in \pgraph$ such that every $B_G$ for $G \in \pgraph$ is logically equivalent to some $B_{G_i}$, where the Boolean formula $B_G$ is defined as $\bigwedge_{\bar{g}} q_{\dcon{G}{\bar{g}}}$. To do this, note that $B_G$ is logically equivalent to $B_H$ if and only if $\{\dcon{G}{\bar{g}} : \bar{g} \in V_G\} = \{\dcon{H}{\bar{h}} : \bar{h} \in V_H\}$ (this is where we use that there is no quantifier alternation). So it is sufficient to prove that, if $c$ is the stated cutoff, 
\[
|V_G|,|V_H| \geq c \implies
\{\dcon{G}{\bar{g}} : \bar{g} \in V_G\} = \{\dcon{H}{\bar{h}} : \bar{h} \in V_H\}
\]




To illustrate how to do this, we analyze the case of uni-directional rings and cliques (the other cases are similar).

\medskip
\noindent
\emph{Uni-directional rings.}
Suppose $\pgraph$ are the uni-directional rings and let $G \in \pgraph$. Fix a $k$-tuple of distinct elements of $V_G$, say $(g_1, g_2, \dots,g_k)$. Define a function $f:V_G \to \{g_1, \dots, g_k\}$ that maps $v$ to the first element of $\bar{g}$ on the path $v, v+1, v+2, \dots$ (addition is mod $|V_G|$). In particular $f(g_i) = g_i$ for $i \in [k]$. In the terminology of the proof of Theorem~\ref{thm:finiteness}, $X(v)$ consists of the simple path $v,v+1,\cdots,f(v)$.

We now describe $\mu_d$.  Clearly $\mu_d(g_i) = \{ \mu_{d-1}(g_i)\}$. By induction on $d$ one can prove that if $v \not\in \bar{g}$ with $f(v) = g_j$ then $\mu_d(v) = \{\mu_{d-1}(v) \cdot \mu_{d-1}(g_j)\}$. So for every $d > 1$, the equivalences $\sim_d$ and  $\sim_1$ coincide.
\begin{center}

\begin{tabular}{>{$}l<{$}|| >{$}c<{$} | >{$}c<{$} | >{$}c<{$} | >{$}c<{$}}
d 		& 0 	& 1 & 2 & \dots \\

\hline
\hline

 \mu_d(v)  \text{ for } v = g_i	&  \{i\}  & \{ \{ i \} \} & \{ \{ \{ i \} \} \} & \dots \\
 \mu_d(v) \text{ if } v \not\in \bar{g} \text{ and } f(v) = g_j 	&  \emptyset  & \{ \emptyset \cdot \{j\} \} &\{  \{ \emptyset \cdot \{j\}\}  \cdot \{\{ j \} \} \} & \dots \\
\end{tabular}
\end{center}
 Thus for every $k \in \Nat$, 
 the  $d$-contraction
   $\dcon{G}{\bar{g}}$ is a ring of size at most $2k$ (in particular, it is independent of $d$). In words, the $d$-contraction of $G$ is the ring resulting by identifying adjacent elements not in $\bar{g}$. It is not hard to see that if $G,H$ are rings such that $|V_G|,|V_H| \geq 2k$ then for every $\bar{g}$ there exists $\bar{h}$ such that $\dcon{G}{\bar{g}} = \dcon{H}{\bar{h}}$.

\todo{compare contraction of rings with EN95}
\medskip
\noindent
\emph{Cliques.}
Fix $n \in \Nat$. Let $G$ be a clique of size $n$. That is: $V_G = [n]$ and $(i,j) \in E_G$ for $1 \leq i \neq j \leq n$. Fix a $k$-tuple of distinct elements of $V_G$, say $(g_1, g_2, \dots,g_k)$.  We now describe $\mu_d(v)$. Clearly $\mu_d(g_i) = \{\mu_{d-1}(g_i)\}$ and for $v \not\in \bar{g}$ we have $\mu_d(v) = \{ \mu_{d-1}(v)\cdot \mu_{d-1}(j) : j \in [k]\}$. So for every $d > 1$, the equivalences $\sim_d$ and  $\sim_1$ coincide, and the $d$-contraction $\dcon{G}{\bar{g}}$ is the clique of size $k+1$. 
In words, the $d$-contraction of $G$ results from $G$ by identifying all vertices not in $\bar{g}$.
 It is not hard to see that  if $G,H$ are cliques such that $|V_G|,|V_H| \geq k+1$ then for every $\bar{g}$ (of size $k$) there exists $\bar{h}$ such that $\dcon{G}{\bar{g}} = \dcon{H}{\bar{h}}$.

%
%
%

\begin{remark} \label{rem:kjb_error}
For cliques and stars, $k+1$ is also a cutoff for $\kCTLstarmX$.
Also, $2k$ is \emph{not} a cutoff for uni-rings and $\kLTLmX$ 
as stated in~\cite[Corollary 2]{KJB13}. To see this, let $tok_i$
express that the process with index $i$ has the token, and $adj(k,i):=
tok_i \rightarrow tok_i \ltlU tok_k 
\lor 
tok_k \rightarrow tok_k \ltlU tok_i$. Then
the formula 
$\exists{i}\exists{j}\forall{k}.\ adj(k,i) \lor adj(k,j)$,
holds in the ring of size $6$, but not $7$.
\end{remark}

\section{There are No Cutoffs for Direction-Aware Systems}

%

In the following, we consider systems where processes can choose which
directions are used to send or receive the token, i.e., process
templates are from $\Psnd$, $\Prcv$, or
$\Prcvsnd$. 
Let \textbf{B} be the
parameterized topology of all bi-directional rings, with directions
${\sf cw}$ (clockwise) and ${\sf ccw}$ (counter-clockwise).
The following theorem answers Question $3$ from the introduction.

\blackout{
\begin{proposition}
For every $\pgraph$, $\PMCP_{\pgraph}(\Prcv,\piLTLmX)$ is decidable 
in case the processes $\Prcv^*$ {\bf sense} where the token comes from, 
but can not choose where to send it.
\end{proposition}
\begin{proof}[Sketch]
Note that the processes can not neither block reception of the message from any direction
nor hold the token forever.

We extend Theorem~\ref{cttv} by redefining the contraction of graph $(G,v_1,\cdots,v_k)$ as follows.
\begin{enumerate}
\item For each pair $v_i, v_j$ and every incoming direction $d$ into $v_j$, add a hub $H = H_{i,j,v}$ with edges $v_i$ to $H$ (labeled with all in-directions) and $H$ to $v_j$ labeled with in-direction $d$ iff there is a path in $G$ from $v_i$ to some $x \neq v_1, \cdots, v_k$ and an edge from $x$ to $v_j$ with in-direction $d$. 
\item For each $v_i$ for which there is an infinite path leaving $v_i$ and never hitting any vertex in $\vec{v}$ again, 
add vertices $X_i,Y_i$ and edges, labeled with all in-directions, from $v_i  \to X_i \leftrightarrow Y_i$.
\end{enumerate}

Every process at a vertex $H_{i,j,d},X_i$ or $Y_i$ works as follows: if it has the token it must give it up; if it does not have the token it receives it with all in-directions.

Obviously a computation in which the token moves from $v_i$ to $v_j$ entering $v_j$ via $d$ can be simulated by pushing the token through hub $H_{i,j,d}$. Conversely, a computation from $v_i$ to $H_{i,j,d}$ to $v_j$ (enters $v_j$ via $d$ by construction) can be simulated by a computation in the original that goes from $v_i$ to $x$ and from $x$ to $v_j$ with in-direction $d$. 

Similarly, a computation in which the token leaves $v_i$ never to hit any vertex in $\vec{v}$ is simulated by pushing the token from $v_i$ to $X_i$ and then to shuttle it between $X_i$ and $Y_i$. Conversely, if the token goes from $v_i$ to $X_i$, then in the original we 

The key point, that distinguishes this case from being able to choose where to send, is that whenever a process gives up the token, any one of its neighbors may be the recipient of the token. Thus: if a path in $G$ exists from $v_i$ to $v_j$ and the token starts with $v_i$ there is always a computation that will result in the token passing through that path.
\end{proof}
}

\begin{theorem} \label{thm:DA}
\begin{enumerate}
\item $\PMCP_{\textbf{B}}(\Prcvsnd, \forall\textsf{-LTL}\backslash\textsf{X})$ is undecidable.
\item For $\mathcal{F}$ equal to $\nineallLTLmX$ or $\nineexistsCTLmX$,
and $\P \in \{\Psnd , \Prcv \}$, 
there exists a parameterized topology $\pgraph$ such that $\PMCP_{\pgraph}(\P, \mathcal{F})$ is undecidable.
\end{enumerate}
\end{theorem}

\paragraph{Proof Idea.} 
We reduce the non-halting problem of two-counter machines
(2CMs) to the PMCP. The idea is that one process, the
\emph{controller}, simulates the finite-state control of the 2CM. The other
processes, arranged in a chain or a ring, are \emph{memory processes},
collectively storing the counter values with a fixed memory per
process. This allows a given system to simulate a 2CM with bounded
counters. Since a 2CM terminates if and only if it terminates for some
bound on the counter values, we can reduce the non-halting problem of
2CMs to the PMCP. The main work is to show that the
controller can issue commands, such as `increment counter $1$' and
`test counter $1$ for zero'. 
We give a detailed proof sketch for part $1$ of the theorem, and then
outline a proof for part $2$.

\noindent \emph{1. $\forall\textsf{-LTL}\backslash\textsf{X}$ and
  $\Prcvsnd$ in bi-directional rings}.

\noindent The process starting with the token becomes the
controller, all others are memory, each storing one bit for each
counter of the 2CM. The current value of a counter $c$ is the total
number of corresponding bits ($c$-bits) set to $1$. Thus, a system with $n$
processes can store counter values up to $n-1$. 

Fix a numbering of 2CM-commands, 
say 
0 $\mapsto$ `increment counter $1$', 
1 $\mapsto$ `decrement counter $1$',
2 $\mapsto$ `test counter $1$ for zero', etc.
Every process has a command variable that represents the command to
be executed when it receives the token from direction {\sf ccw}.

If the controller sends the token in direction {\sf cw}, the memory
processes will increment (mod $6$) the command variable, allowing the
controller to encode which command should be executed. Every process
just continues to pass the token in direction {\sf cw}, until it
reaches the controller again.


If the controller sends the token in direction {\sf ccw}, then the 
memory processes try to execute the command currently
stored. If it is an 'increment counter $c$' or 'decrement counter $c$'
command, the memory process tries to execute it (by
incrementing/decrementing its $c$-bit). If the process cannot execute
the command (because the $c$-bit is already $1$ for an increment, or
$0$ for a decrement), then it passes the token along direction {\sf ccw} and remembers that a
command is being tried. If the token reaches a memory process which can
execute the command, 
then it does so and passes the token back in
direction {\sf cw}. The processes that remembered that a command is
being tried will receive the token from direction {\sf cw}, and know
that the command has been successfully executed, and so will the
controller. 
 If the controller gets the token from {\sf ccw}, the command
failed. In this case, the controller enters a loop in which it
just passes the token in direction {\sf cw} (and no more commands are
executed).

If the command stored in the memory processes is a `test for zero
counter $c$', then the processes check if their $c$-bit is $0$. If
this is the case, it (remembers that a command is being tried and)
passes the token to the next process in direction {\sf ccw}. If the
token reaches a process for which the $c$-bit is $1$, then this
process sends the token back in direction {\sf cw}. Other memory
processes receiving it from {\sf cw} (and remembering that the command
is being tried), pass it on in direction {\sf cw}. In this case, the
controller will receive the token from {\sf cw} and know that counter
$c$ is not zero. On the other hand, if all memory processes store $0$
in their $c$-bit, then they all send the token in direction {\sf
  ccw}. Thus, the controller will receive it from {\sf ccw} and knows
that counter $c$ currently is zero. To terminate the command, it sends
the token in direction {\sf cw}, and all processes (which remembered
that a command is being tried), know that execution of this command is
finished.\sr{should probably point out that if test for zero failed
  then all processes have their pointer at the same command.} 

With the description above, a system with $n-1$ memory processes can
simulate a 2CM as long as counter values are less than $n$. Let $HALT$
be an atomic proposition that holds only in the controller's halting
states. Then
solving the PMCP for $\forall{i} \ltlG\neg HALT_i$ 
amounts to solving the non-halting problem of the 2CM.

\noindent \emph{2. $\nineallLTLmX$ and $\Psnd$}.

\noindent We give a proof outline. In this case
there are $2n$ memory processes, $n$ for each counter $c \in
\{1,2\}$. The remaining $9$ processes are special and called
`controller', 'counter $c$ is zero', `counter $c$ is not zero',
`counter $c$ was incremented', and `counter $c$ was decremented'. When
the controller wants to increment or decrement counter $c$, it sends
the token non-deterministically to some memory process for counter
$c$. When the controller wants to test counter $c$ for zero, it sends
the token to the first memory process. When a memory process receives
the token it does not know who
sent it, and in particular does not know the intended command. Thus,
it non-deterministically takes the appropriate action for one of the
possible commands. If its bit is set to $0$ then it either i) increments
its bit and sends the token to a special process `counter $c$ was
incremented', or ii) it sends the token to the next memory node in the
chain, or to the special process `counter $c$ is zero' if it is the
last in the chain. If its bit is set to $1$ then it either i) decrements its
bit and sends the token to a special process `counter $c$ was
decremented', or ii) sends the token to a special process `counter $c$
is not zero'. 

Even though incoming directions are not available to the processes, we
can write the specification such that, out of all the possible
non-deterministic runs, we only consider those in which the controller
receives the token from the expected special node (the
formula requires one quantified index variable for each of the special
nodes). So, if the controller wanted to increment counter $c$ it needs
to receive the token from process 'counter $c$ was incremented'. If
the controller receives the token from a different node, it means that
a command was issued but not executed correctly, and the formula disregards this run. Otherwise, the system of size $2n+1$ correctly simulates the 2CM
until one of the counter values exceeds $n$. 
\qed

\blackout{
\todo{old sketch}
\begin{proof}
Given a 2CM we build a topology $\pgraph$ (on Fig.\ref{fig:fig_piltl_cn_undec}), 
a process template $P \in \Psnd$,
and a formula $F$ of the form 
$\forall x \forall y_1 \cdots \forall y_8: \,  \ltlG\ValidRun \to \ltlG\neg HALT_x$, 
such that the given 2CM halts if and only if the answer to 
$\PMCP_\pgraph(\{P\},\{F\})$ is `No'.

On first receiving the token a process non-deterministically transitions into one of 11 components:
`controller' (at vertex $con$) that simulates the control of the 2CM, 
and for $c\in[2]$,
`memory $c$' (at vertices $m^c_i$) that stores the value of counter $c$, 
helpers $inc^c,dec^c,nz^c,z^c$ that just passes the token to $con$.
A special initialization phase together with formula $F$ ensures that 
a process at vertex $con$ is in `controller' component, 
processes at vertices $m^c_i$ are in `memory $c$' component, etc.

The `memory $c$' component 
stores a binary value $b$ that starts with $b = 0$.
If a process has $b = 0$ and has the token then the process can either
\begin{itemize}
\item send the token along the $INC^c$ edge and set $b = 1$, or 
\item send the token along the $M^c$ edge,  or 
\item in case the process is $m^c_n$ it can send the token along the $Z^c$ edge,
\end{itemize} 
and if a process has $b = 1$ and holds the token then the process can either 
\begin{itemize}
\item send the token along a $DEC^c$ edge and set $b = 0$, or 
\item send the token along the $NZ^c$ edge.
\end{itemize}

%
\noindent The controller process simulates the 2CM:
\begin{itemize}
\item If the command of the 2CM is \textit{`increment $c$'} 
then the controller sends the token along $M^c$ edge and waits 
for the token from edge $inc^c$. 
Although the controller cannot choose the token receiving direction, 
the formula part $\ltlG \ValidRun$ filters out runs with other receiving directions.


\item If the command is
\textit{`test $c$ for zero'}, 
then the controller sends the token along $TZ^c$ edge, 
and then guesses the result:
if it guessed `$c$ is zero', then it waits the token from $z^c$ node, 
otherwise -- from $nz^c$.
The formula part $\ltlG\ValidRun$ filters out runs with wrong guesses.
\end{itemize}

\iffig
\begin{figure}
 \begin{subfigure}[b]{0.49\textwidth}
    \centering
      \begin{tikzpicture}[->,
node distance = 1.5cm, 
auto, 
semithick, 
inner sep=.01cm,
bend angle=45]
    \tikzset{every state/.style={rectangle,rounded corners,minimum size=2em}}
    \tikzset{every edge/.append style={font=\small, right}}
    \tikzset{box state/.style={draw,rectangle,rounded corners,inner sep=.1cm}}
    
    \tikzset{SnakeStyle/.style = {snake,segment amplitude=.2mm,segment length=1mm,line after snake=2mm}}
    
    \node[state] (con) {$con$};
    \node[state] (inc) [below right of=con] {$inc^c$};
    \node[state] (dec) [right of=inc] {$dec^c$};
    \node[state] (nz) [right of=dec] {$nz^c$};
    \node[state] (z) [right of=nz] {$z^c$};
    
    \path (dec)  edge (con);
    \path (inc)  edge (con);
    \path (nz)  edge (con);
    \path (z)  [bend right=5] edge (con);
    
    \node[state] (m0) [below=2cm of con] {$m_1^c$};
    \node[box state] (dots) [below=1cm of m0] {$\ldots$};
    \node[state] (mn) [below=1cm of dots] {$m_n^c$};
    
    \path (con)  edge node {$TZ^c$} (m0);
    \path (con)  [bend right=35] edge (m0);
    \path (con)  [dotted, bend right=35] edge (dots);
    \path (con)  [bend right=35] edge node [near start] {$M^c$} (mn);

    \path (m0) edge node {$INC^c$} (inc);
    \path (m0) [bend right=15] edge (dec);
    \path (m0) [bend right=15] edge [near end, right] node {$DEC^c$} (dec);
    \path (m0) [bend right=15] edge [near end, right] node {$NZ^c$} (nz);
    \path (m0) edge [near start, left] node {$M^c$} (dots);
    
    \path (dots) edge [near start, left] node {$M^c$} (mn);
    
    \path (dots) [dotted] edge (inc);
    \path (dots) [dotted, bend right=15] edge (dec);
    \path (dots) [dotted, bend right=15] edge (nz);
    
    \path (mn) edge (inc);
    \path (mn) [bend right=15] edge (dec);
    \path (mn) [bend right=15] edge (nz);
    \path (mn) [bend right=15] edge node [near end] {$Z^c$} (z);
    
\end{tikzpicture}
    \caption{Edges used in the simulation of a 2CM.}
    \label{fig:fig_piltl_ny_undec_sim}
 \end{subfigure}
 \begin{subfigure}[b]{0.49\textwidth}
    \centering
      \begin{tikzpicture}[->,
node distance = 1.5cm, 
auto,
semithick, 
inner sep=.01cm,
bend angle=45]
    \tikzset{every state/.style={rectangle,rounded corners,minimum size=2em}}
    \tikzset{every edge/.append style={font=\small, right}}
    
    \tikzset{SnakeStyle/.style = {snake,segment amplitude=.2mm,segment length=1mm,line after snake=1mm}}
    
    \node[state] (con) {$con$};
    \node[state] (inc) [below of=con] {$inc^c$};
    \node[state] (dec) [right of=inc] {$dec^c$};
    \node[state] (nz) [right of=dec] {$nz^c$};
    \node[state] (z) [above of=nz] {$z^c$};
    

    \path (con)  edge[bend right=10] (inc);
    \path (con)  edge[bend right=5] (dec);
    \path (con)  edge[bend right=5] (nz) ;
    \path (con)  edge[bend right=5] (z)  ;

    \path (inc)  edge[bend right=10] (con);
    \path (dec)  edge[bend right=5] (con);
    \path (nz)   edge[bend right=5] (con);
    \path (z)    edge[bend right=5] (con);

\end{tikzpicture}
    \caption{Edges used in the initialization phase.}
    \label{fig:fig_piltl_ny_undec_init}
 \end{subfigure}
\caption{Topology $G$ used in Theorem~\ref{thm:NC}($c\in [2]$, edges to the `sink' are not shown)}
\label{fig:fig_piltl_ny_undec}
\end{figure}
\fi
\ak{use different colors for edges of different type.}
\ak{add labels to edges on fig(b)}

Intuitively, ``$\ltlG\ValidRun$ is true'' 
means that during the run almost all guesses of the processes were correct: 
the controller correctly guessed results of 
\textit{`test $c$ for zero'} commands,
and there was a memory process who correctly guessed the command
and sent the token to one of $inc^c,dec^c,nz^c,z^c$ 
(although some of the memory processes could wrongly guess the 
\emph{`test $c$ for zero'} command).

Consider how a simulation of an \emph{`increment $c$'} command works.
The controller sends the token along $M^c$ edge 
to a non-deterministically chosen memory node.
If the memory process has its $b=0$, 
then it guesses the command to execute: 
\emph{`increase $c$'} or \emph{`test $c$ for zero'}.  
If it guesses the former, then it sets $b=1$ and sends the token to $inc^c$ 
completing the command. 
If it guesses \emph{`test $c$ for zero'} then it sends the token along $M^c$ edge.
It is still possible to complete the command if the token end up 
in the memory process with $b=0$ that correctly guesses the command.
Otherwise, the token will be sent to one of $dec^c,nz^c,z^c$
nodes and $\ltlG \ValidRun$ will be falsified. 

Similarly works simulation of \emph{`test $c$ for zero'}. If the controller
guesses `$c$ is zero', then in order $\ltlG \ValidRun$ to be true, 
the token should go through all of the memory nodes, which is possible 
only if all the memory nodes have $b=0$ and correctly guess the command.
Note that the memory process having $b=1$ never sends the token along 
$M^c$ edge to the next memory process. \qed
\end{proof}
}

\blackout{
\begin{theorem}\label{thm:CN}
There exists a parameterized topology $\pgraph$ such that the problem $\PMCP_{\pgraph}(\Prcv,\piLTLmX)$ is undecidable.
\end{theorem}
\begin{proof}
The idea is similar to $\Psnd$ case, see Appendix~\ref{appendix:dir-aware}.
\qed
\end{proof}


\begin{theorem} \label{thm:CCbirings}
The problem $\PMCP_{\textbf{BR}}(\Prcvsnd, \oneSLTLmX)$ is undecidable.
\end{theorem}\todo{introduce 1-SLTL in def of prenex indexed formulas}
\begin{proof}[full proof is in Appendix~\ref{appendix:dir-aware}]
The process having the token initially simulates 2CM, 
while others store counters values. 
The controller conveys the command to memory processes 
by sending the token in ${\sf cw}$ 
direction a special number of times. The controller signals 
the memory processes to execute the command by sending the token 
in ${\sf ccw}$ direction. \qed
\end{proof}
}

\blackout{
\begin{question}
Can we characterize for which graphs $\PMCP$ with directions is decidable. 
Graphs of bounded diameter?
\end{question}
}

%

\section{Extensions} \label{sec:extensions} 

There are a number of extensions of direction-unaware TPSs for which
the theorems that state existence of cutoffs (Theorems~\ref{thm:cutoff_kctld} and \ref{thm:explicit_cutoffs}) still hold. We describe these in order to highlight assumptions that make the proofs work:
\begin{enumerate}
\item Processes can be infinite-state. 
\item The EN-restriction on the process template $P$ can be relaxed: replace item $vii)$ in Definition~\ref{def:process_template} by 
``For every state $q$ that has the token there is a finite path 
$q \dots q'$ such that $q'$ does not have the token, 
and for every $q$ that does not have the token there is a finite path
$q \dots q'$ such that $q'$ has the token''.

\item One can further allow \emph{direction-sensing} TPSs,
which is a direction-aware TPS with an additional restriction 
on the process template: 
``If $\trans{q}{q'}{d} \in \Trans$ for some direction $d \in \dirset_\tsnd$,
then for every $d \in \dirset_\tsnd$ there exists 
a transition $\trans{q}{q''}{d} \in \Trans$''; 
and a similar statement for $\dirset_\trcv$. Informally: 
we can allow processes to change state according to the direction 
that the token is (non-deterministically) sent to or received, 
but the processes are not allowed to block any particular direction.
\ak{Did you check this item is allowed?}

\item One can further allow the token to carry a value but with the
  strong restriction that from every state that has the token and
  every value $v$ there is a path of internal actions in $P$ which
  eventually sends the token with value $v$, and the same for receiving.
\sr{This should be checked. It can wait for the journal version.}

\end{enumerate}
These conditions on $P$ all have the same flavor: they ensure that a process can not choose what information to send/receive, whether that information is a value on the token or a direction for the token.

\section{Related work}\label{sec:related} 


Besides the results that this paper is directly based on~\cite{Suzuki,EN95,CTTV04}, there are several other relevant papers.

Emerson and Kahlon~\cite{EK04} consider token-passing in uni- and bi-directional rings, where processes are direction-aware and tokens carry messages (but can only be changed a bounded number of times). However, the provided cutoff theorems only hold for specifications that talk about two processes (in a uni-directional ring) or one process (in a bi-directional ring), process templates need to be deterministic, an cutoffs depend on the size of the process implementation.

German and Sistla~\cite{GS92} provide cutoffs for the PMCP for systems with pairwise synchronization. Although pairwise synchronization can simulate token-passing, their cutoff results are restricted to cliques and $1$-indexed \LTL. Moreover, their proof uses vector-addition systems with states and their cutoff depends on the process template and the specification formula. 



Delzanno et al.~\cite{DelzannoSZ10} 
study a model of broadcast protocols on arbitrary topologies, 
in which a process can synchronize with all of its available neighbors `at once' 
by broadcasting a message (from a finite set of messages).
They prove undecidability of PMCP for systems with arbitrary topologies 
and $1$-indexed safety properties, and that the problem becomes decidable if one restricts the topologies to `graphs with bounded paths' (such as stars). 
Their proof uses the machinery of well-structured transitions systems,
and no cutoffs are provided.
They also show undecidability of the PMCP in the case of non-prenex indexed properties 
of the form $\ltlG(\exists{i}.s(i) \in B)$
on general and the restricted topologies. 

Rabinovich~\cite[Section $4$]{rabi07} proves, using the composition method, that if monadic second-order theory of the set of topologies in $\pgraph$ is decidable, then the PMCP is decidable for propositional modal logic. The systems considered are defined by a very general notion of product of systems (which includes our token passing systems as a subcase).

The PMCP for various fragments of non-prenex indexed \LTL\ is
undecidable, see German and Sistla~\cite[Section $6$]{GS92} for
systems with pairwise synchronization, and John et al.~\cite[Appendix
  $A$]{igor12} for systems with no synchronization at all.

\blackout{
Rabinovich~\cite{rabi07} defines a very general notion of product of systems (which includes our token passing systems as a subcase) and proves a composition theorem in the spirit of Feferman-Vaught \cite{FV59} for specifications from basic propositional (multi-) modal logic. He also briefly considers the parameterized model checking problem \cite[Section $4$]{rabi07} and proves that if a parameterized topology $\pgraph$ is algorithmically nice (i.e., has decidable monadic second-order theory in the signature of graphs) \rb{what does it mean ``the parameterized topologies $\pgraph$ have decidable monadic second-order theory''? does it mean the topology can be described in MSO?} then PMCP is decidable for propositional modal logic. Note that his definition of product of systems can introduce new atomic predicates of the form $\forall i. p_i$ \cite[Example $20. (1)$]{rabi07} (and thus can form formulas that are not in the prenex fragment of indexed temporal logic), as well as monadic second-order quantification (e.g., in rings one can express that there are an even number of $i$ such that $p_i$). 
The exact relationship between the resulting specification language and (full, not necessarily prenex) indexed-temporal logic is left for future work.
\rb{future work of whom? us? Rabinovich?}
}

\section{Summary}

The goal of this work was to find out under what conditions there are
cutoffs for temporal logics and token-passing systems on general
topologies. We found that stratifying prenex indexed \CTLstarmX\ by
nesting-depth of path quantifiers allowed us to recover the existence
of cutoffs; but that there are no cutoffs if the processes are allowed
to choose the direction of the token. 
In all the considered cases
where there is no cutoff we show that the PMCP problem is actually
undecidable.\sr{im not sure how to include the questions here...}

Our positive results are provided by a construction that generalizes
and unifies the known positive results, and clearly decomposes the
problem into two aspects: tracking the movement of the token through
the underlying topology, and simulating the internal states of the
processes that the specification formula can see.  The construction
yields small cutoffs for common topologies (such as rings, stars, and
cliques) and specifications from prenex indexed \CTLstarmX. 

\blackout{

The following is not in the paper but can be dealt with.
\begin{list}{$-$}{}
\item the weaker notion of direction-sensing in which a process must receive the token but can change its state depending on the direction it came from. can recover cutoffs exist.
\item two other notions of fairness: CTTV (global fairness in which you ignore non fair runs), and the condition that allows you always to push the token through.
\item multiple identical tokens
\item binary value with suitable restriction.
\item decidable QATL fragment, and QATL and prenexTL are expressively incomparable.
\item the cutoffs exist proofs work (except for \CTLmX) for infinite processes P.
\end{list}

\begin{remark}
The proof (of Theorem 8) also shows that there exists finitely many topologies $H_1,\cdots,H_m$ such that if $P^{H_i} \models \phi$ for $1 \leq i \leq m$ then $P^G \models \phi$ for almost all\footnote{i.e., for all except possibly a finite number.}\ak{why except possible a finite number?}  $G \in \pgraph$.
\end{remark}

{\bf Atom-quantified temporal logic (AQTL)} is the fragment of ITL in which quantification can only apply to Boolean combinations of atoms, e.g., $\eventually \exists i p_i$. 
Another popular fragment is \emph{quantified-atom temporal logic, QATL} in which there are no temporal operators in the scope of quantifiers.

We consider \LTLmX as the syntactic fragment of \CTLstarmX\ in which there is a single path quantifier, namely $\ctlA$, at the front of the formula.

\subsection{Open problems.} 
\todo{find a place for open problems... conclusion?}
The following problems are left open:

\begin{question}
\ak{this question makes sense for each constant $k$, 
not in general.}
If processes know to whom they send tokens, but not from where they
receive them, is $\PMCP$ for $k-\LTLmX$ decidable? What is a suitable
decomposition of a big graph?
\end{question}
Decompositions need to take into account outgoing directions, and may
become significantly more complex.

\begin{question}
Is $\PMCP(\Pproctemp_{simptok}, \{biring^{recv'}\}, \Pspec)$ decidable?
(if a process can sense but cannot control the direction) 
\end{question}

\begin{question}
Is $\PMCP(\Pproctemp_{simptok}, \{biring^{send'}\}, \Pspec)$ decidable?
(if a process can control but cannot sense the direction) 
\end{question}

}

\smallskip
\noindent
{\bf Acknowledgments.}
%
We thank Roderick Bloem for detailed comments on numerous drafts and Krishnendu Chatterjee for important comments regarding the structure of the paper. We thank Roderick Bloem, Igor Konnov, Helmut Veith, and Josef Widder for discussions at an early stage of this work that, in particular, pointed out the relevance of direction-unawareness in \cite{CTTV04}.

\bibliographystyle{splncs03}
\bibliography{lit}



\end{document}